\documentclass[review,1p,times,twocolumn]{elsarticle}

\usepackage{lineno}
\modulolinenumbers[5]

\usepackage{amsfonts}
\usepackage{amsmath}
\usepackage{amssymb}
\usepackage{amsthm}
\usepackage{setspace}
\usepackage{graphicx}
\usepackage{afterpage}
\usepackage[ruled,titlenumbered,lined]{algorithm2e}
\usepackage{varioref}
\usepackage{hyperref}
\usepackage{bookmark}
\usepackage{pdflscape}
\usepackage{colortbl,color}
\usepackage{subcaption}
\usepackage{soul}
\usepackage{pgf, tikz}
\usetikzlibrary{arrows, automata}
\usepackage{caption}
\usepackage{longtable}

\usepackage{textcomp}
\usepackage{siunitx}
\usepackage{booktabs}
\usepackage{multirow}
\usepackage{xcolor}
\newtheoremstyle{exampstyle}
  {\topsep} % Space above
  {\topsep} % Space below
  {} % Body font
  {} % Indent amount
  {\bfseries} % Theorem head font
  {:} % Punctuation after theorem head
  {.5em} % Space after theorem head
  {} % Theorem head spec (can be left empty, meaning `normal')

\allowdisplaybreaks[1]

\newtheorem{theorem}{Theorem}
\newtheorem{proposition}{Proposition}
\newtheorem{corollary}{Corollary}

\makeatletter
\newcommand*{\rom}[1]{\expandafter\@slowromancap\romannumeral #1@}
\makeatother

\journal{arXiv.org}

\begin{document}

\begin{frontmatter}

%% Title, authors and addresses

%% use the tnoteref command within \title for footnotes;
%% use the tnotetext command for theassociated footnote;
%% use the fnref command within \author or \address for footnotes;
%% use the fntext command for theassociated footnote;
%% use the corref command within \author for corresponding author footnotes;
%% use the cortext command for theassociated footnote;
%% use the ead command for the email address,
%% and the form \ead[url] for the home page:
%% \title{Title\tnoteref{label1}}
%% \tnotetext[label1]{}
%% \author{Name\corref{cor1}\fnref{label2}}
%% \ead{email address}
%% \ead[url]{home page}
%% \fntext[label2]{}
%% \cortext[cor1]{}
%% \address{Address\fnref{label3}}
%% \fntext[label3]{}

\title{A cutting-plane algorithm for the Steiner team orienteering problem}

%% use optional labels to link authors explicitly to addresses:
%% \author[label1,label2]{}
%% \address[label1]{}
%% \address[label2]{}

\author[dep]{Lucas Assun\c{c}\~ao\corref{cor1}}
\ead{lucas-assuncao@ufmg.br}
\author[dcc]{Geraldo Robson Mateus}
\ead{mateus@dcc.ufmg.br}
\cortext[cor1]{Corresponding author.}

\address[dep]{
  Departamento de Engenharia de Produ\c{c}\~ao, Universidade Federal de Minas Gerais\\
  Avenida Ant\^onio Carlos 6627, CEP 31270-901, Belo Horizonte, MG, Brazil\\
}

\address[dcc]{
  Departamento de Ci\^encia da Computa\c{c}\~ao, Universidade Federal de Minas Gerais\\
  Avenida Ant\^onio Carlos 6627, CEP 31270-901, Belo Horizonte, MG, Brazil\\
}

\begin{abstract}
The Team Orienteering Problem (TOP) is an NP-hard routing problem in which a fleet of identical vehicles aims at collecting rewards (prizes) available at given locations, while satisfying restrictions on the travel times. In TOP, each location can be visited by at most one vehicle, and the goal is to maximize the total sum of rewards collected by the vehicles within a given time limit. In this paper, we propose a generalization of TOP, namely the Steiner Team Orienteering Problem (STOP). In STOP, we provide, additionally, a subset of mandatory locations. In this sense, STOP also aims at maximizing the total sum of rewards collected within the time limit, but, now, every mandatory location must be visited.
In this work, we propose a new commodity-based formulation for STOP and use it within a cutting-plane scheme. The algorithm benefits from the compactness and strength of the proposed formulation and works by separating three families of valid inequalities, which consist of
some general connectivity constraints, classical lifted cover inequalities based on dual bounds and a class of conflict cuts. To our knowledge, the last class of inequalities is also introduced in this work.
A state-of-the-art branch-and-cut algorithm from the literature of TOP is adapted to STOP and used as baseline to evaluate the performance of the cutting-plane.
Extensive computational experiments show the competitiveness of the new algorithm while solving several STOP and TOP instances. In particular, it is able to solve, in total, 14 more TOP instances than any other previous exact algorithm and finds eight new optimality certificates. With respect to the new STOP instances introduced in this work, our algorithm solves 30 more instances than the baseline.
\end{abstract}

\begin{keyword}
  Vehicle routing \sep Orienteering problems \sep Cutting-plane
%% PACS codes here, in the form: \PACS code \sep code

%% MSC codes here, in the form: \MSC code \sep code
%% or \MSC[2008] code \sep code (2000 is the default)

\end{keyword}

\end{frontmatter}
%\linenumbers

\section{Introduction}

Orienteering is a sport usually practiced in places with irregular terrains, such as mountains and dense forests.
It is given a set of \emph{control points} to be visited, each of them with an associated reward (prize).
The competitors are provided a topographical map and compass in order to guide them from an origin point to a destination one,
which are the same for all competitors.
Their goal is to maximize the total sum of rewards collected from visiting the control points within a previously established
time limit. Each reward can be collected by a single competitor, and the winner is the one that reaches the destination point within the time limit with the maximum amount
of rewards.

Based on this sport, Tsiligirides introduced the \emph{Orienteering Problem} (OP)~\cite{Tsiligirides84} . The problem is defined on a graph, usually complete and undirected, where a value of reward is associated with each vertex and a traverse time is associated with each edge (or arc).
OP aims at finding a route from an origin vertex to a destination one (visiting each vertex at most once) that
satisfies a total traverse time constraint while maximizing the sum of rewards collected. In OP, a reward cannot be multiply collected, just like in the original orienteering sport. In fact, an optimal route for OP corresponds to an optimal one for an orienteering competitor, except for the fact that, in OP, no vertex can be multiply visited. 
We also point out that, contrary to the classical
\emph{Traveling Salesman Problem} (TSP) \cite{Dantzig54}, a solution for OP does not necessarily visit all the vertices of the graph.

When the origin and the destination vertices coincide, the problem is known as the \emph{selective traveling salesman problem}
\cite{Laporte90}. Moreover, when we consider a team of competitors working together, the problem becomes the 
Team Orienteering Problem (TOP) \cite{Chao96}. In TOP, all the $m$ members of the team depart from the same vertex at the same time
and have to arrive at the destination vertex, also within a same time limit.
The goal is to find $m$ routes that, together, maximize the total reward collected by the team. As for OP, a vertex/reward cannot
be multiply visited/collected, i.e., once a member of the team collects the reward of a vertex, this vertex cannot be visited 
again.

Both OP and TOP are NP-hard \cite{Laporte90,Poggi10} and find applications in transportation and delivery of goods \cite{Viana11}.
With the advent of the \emph{e-commerce}, for instance, several virtual stores assign to different shipping companies
their delivery requests.
Nevertheless, the fleet available to a given shipping company is not always enough to perform all the deliveries assigned to it in a single working
day. In these cases, the company must select only a subset of the total amount of its deliveries.
To this end, a value of priority can be associated with each delivery. This value corresponds to the reward 
achieved by performing the delivery in the current working day and might combine different factors, such as
the importance of the client and the urgency of the request.

A similar application arises in the planning of technical visits \cite{Tang05}.
Also in this case, a reward is associated with visiting each customer and performing a given service.
Likewise, the values of the rewards rely on factors such as the priority of the service and the importance of the customer.
Therefore, the goal is to select a subset of technical visits (to be performed within a working horizon of time)
that maximizes the total sum of rewards achieved.
Notice that, in both applications, this priority policy is not enough to ensure that deliveries or technical
visits with top priority (e.g., those whose deadlines are expiring) will be necessarily selected in the planning.
In this study, we propose a variation of TOP, namely the \emph{Steiner Team Orienteering Problem} (STOP), that addresses this issue.

STOP is defined on a digraph, where an origin and a destination vertices are given, and the remaining vertices are subdivided into two categories: the mandatory ones, which must necessarily be visited, and the profitable ones, which work as Steiner vertices and, thus, may not be visited. A traverse time is associated with each arc in this digraph, and values of reward are associated with visiting the profitable vertices. In order to represent the team of members, it is also given a homogeneous fleet of vehicles, which can only run for a given time limit.
STOP aims at finding routes (one for each vehicle) from the origin vertex to the destination one such that every mandatory vertex belongs to exactly one route and the total sum of rewards collected on the visited profitable vertices is maximized. Here, each profitable vertex can be visited by at most one vehicle, thus avoiding the multiple collection of a same reward.

The main contribution of this paper consists of introducing a commodity-based compact formulation for STOP and devising a cutting-plane scheme to solve it. The cutting-plane relies on the separation of three families of valid inequalities, which consist of some general connectivity constraints, classical lifted cover inequalities based on dual bounds and a class of conflict cuts. As far as we are aware, the last class of inequalities is also introduced in this work and can be applied to similar problems in a straightforward manner.
Our algorithm highly benefits from the compactness and the strength of the formulation proposed, which we prove to give the same bounds as the one used within a state-of-the-art branch-and-cut algorithm from the literature of TOP. In this work, we adapt to STOP this branch-and-cut algorithm and use it as a baseline to evaluate the performance of the cutting-plane proposed. According to extensive experiments, our algorithm shows to be highly competitive with previous exact approaches in the literature. In fact, it is able to solve, in total, 14 more instances than any other TOP exact algorithm and finds the optimality certificates of eight previously unsolved TOP instances. With respect to the new STOP instances introduced in this work, the new algorithm solves to optimality 30 more instances than the baseline. 

The remainder of this work is organized as follows. Related works are discussed in Section~\ref{s_related_works}, and STOP is formally defined in Section~\ref{s_notation}.
In Section~\ref{s_models}, we present the two formulations used as backbone of the exact algorithms developed. In the same section, these formulations are also proven to provide the same bounds. Section~\ref{s_cuts} is devoted to describing the three families of valid inequalities used within the cutting-plane scheme proposed. Moreover, the procedures used to separate these inequalities are presented in Section~\ref{s_separation}. The baseline branch-and-cut algorithm adapted to STOP is briefly described in Section~\ref{s_branch-and-cut}, and the cutting-plane scheme proposed is detailed in Section~\ref{s_cutting-plane}. Some implementation details are given in Section~\ref{s_implementation_details}, followed by computational results (Section~\ref{s_experiments}). Concluding remarks are provided in the last section. 

\section{Related works}
\label{s_related_works}

Although STOP has not been addressed in the literature yet, a specific case of the problem that
considers a single vehicle, namely the \emph{Steiner Orienteering Problem} (SOP), was already introduced by Letchford et al.~\cite{Letchford13}.
In the work, the authors propose four Integer Linear Programming (ILP) models for the problem, but no computational
experiment is reported. STOP is also closely related to several routing problems, such as TOP, OP and the Capacitated Vehicle Routing Problem (CVRP) \cite{Toth2001} and its variations.
In the remainder of this section, we present a literature review on the main heuristic and exact algorithms to solve TOP, the problem most closely related to STOP.

The particular case of STOP with no mandatory vertices, namely TOP, was introduced by the name of the
\emph{multiple tour maximum collection problem} in the work of Butt and Cavalier~\cite{Butt94}. Nevertheless, the problem was only
formally defined by Chao et al.~\cite{Chao96}. In the latter work, the first TOP instance set was proposed, along with a heuristic procedure.

Throughout the last decade, several heuristics have been proposed for TOP. For instance, Tang and Miller-Hooks~\cite{Tang05} presented an algorithm that combines a \emph{Tabu Search} (TS) heuristic with an adaptive memory
procedure. Archetti et al.~\cite{Archetti07} developed two more TS heuristics for the problem, as well as two procedures based on
\emph{Variable Neighborhood Search} (VNS). In addition, Ke et al.~\cite{Ke08} proposed \emph{ant colony} based algorithms which presented
results comparable to those of \cite{Archetti07}, with less computational time effort.
The VNS heuristic of Vansteenwegen et al.~\cite{Vansteenwegen09} was the first procedure to focus on time efficiency. However, the quality
of the solutions obtained by it is slightly worse than that of the solutions obtained by \cite{Archetti07}.

Later on, in \cite{Souffriau10}, the same authors of \cite{Vansteenwegen09} proposed a
\emph{Greedy Randomized Adaptive Search Procedure} (GRASP) metaheuristic with \emph{Path Relinking} (PR) which was able
to outperform all the heuristic approaches aforementioned \cite{Chao96,Tang05,Archetti07,Ke08,Vansteenwegen09}
both in effectiveness (i.e., bounds of the solutions) and time efficiency.
More recently, three new approaches were able to outperform the results of \cite{Souffriau10}: the
\emph{Simulated Annealing} (SA) heuristic of Lin~\cite{Lin13}, the \emph{Large Neighborhood Search} (LNS) based heuristics of Kim et al.~\cite{Kim13}
and the evolutionary algorithm of Dang et al.~\cite{Dang13}, which is inspired by \emph{Particle Swarm Optimization} (PSO).

The algorithm of Dang et al.~\cite{Dang13} showed to be competitive with the ones of Kim et al.~\cite{Kim13} in terms of the quality of the solutions
obtained for complete digraph instances with up to 100 vertices. However, according to the results, the latter heuristics \cite{Kim13}
are more efficient.
Dang et al.~\cite{Dang13} also tested their evolutionary algorithm on larger instances, with up to 400 vertices. Due to the lack
of optimality certificates for these instances, the heuristic was only evaluated in terms of stability and time
efficiency in these cases. The results obtained by Lin~\cite{Lin13} were not compared to those of Kim et al.~\cite{Kim13} and Dang et al.~\cite{Dang13}.

To our knowledge, the latest heuristic for TOP was proposed by Ke et al.~\cite{Ke2016}.
Their heuristic, namely \emph{Pareto mimic algorithm}, introduces a so-called \emph{mimic operator} to generate new solutions
by imitating incumbent ones. The algorithm also adopts the concept of \emph{Pareto dominance} to update the population of
incumbent solutions by considering multiple indicators that measure the quality of each solution. The results indicate that
this new algorithm can achieve all the best-known bounds obtained by Lin~\cite{Lin13} and Dang et al.~\cite{Dang13}. In addition,
the algorithm of Ke et al.~\cite{Ke2016} was even able to find improved bounds for 10 of the larger instances (with up to 400 vertices) introduced by Dang et al.~\cite{Dang13}.

Although there are several heuristics to solve TOP, only a few works propose exact solution approaches for the problem.
As far as we are aware, Butt and Ryan~\cite{Butt99} presented the first exact algorithm for TOP, which is based on column generation.
More recently, Boussier et al.~\cite{Boussier07} proposed a \emph{set packing} formulation with an exponential number of variables, 
each of them representing a feasible route. In the work, the formulation is solved by means of a branch-and-price
algorithm, and the pricing sub-problems are solved through dynamic programming.
In \cite{Poggi10}, Poggi et al. proposed a branch-and-cut-and-price algorithm,
along with new min-cut and triangle clique inequalities. The algorithm solves a Dantzig-Wolfe reformulation of
a pseudo-polynomial compact formulation where edges are indexed by the time they are placed in a route.

Later on, Dang et al.~\cite{Dang13b} developed a branch-and-cut algorithm that relies on a set of dominance properties and valid
inequalities, such as symmetry breaking, generalized sub-tour eliminations and clique cuts based on graphs of
incompatibilities. The algorithms of Poggi et al.~\cite{Poggi10} and Dang et al.~\cite{Dang13b} were both able to obtain new optimality certificates.
Moreover, the branch-and-cut algorithm of Dang et al.~\cite{Dang13b} showed to be competitive with the
branch-and-price algorithm of Boussier et al.~\cite{Boussier07}. Since the authors of \cite{Poggi10} do not report the experimental results for the whole
benchmark of TOP instances in the literature, the performance of their algorithm could not be properly compared with other
approaches.

Recently, Keshtkaran et al.~\cite{Keshtkaran16} proposed a branch-and-price algorithm where the pricing sub-problems are solved be means of a dynamic
programming algorithm with decremental state space relaxation featuring a two-phase dominance rule relaxation.
The authors also presented a branch-and-cut-and-price algorithm that incorporates a family of subset-row inequalities
to the branch-and-price scheme.
The two algorithms showed to be competitive with the previous exact methods in the literature. In fact, they both were able to outperform the algorithms of Boussier et al.~\cite{Boussier07} and Dang et al.~\cite{Dang13b} in terms of the total number of instances
solved to proven optimality within the same execution time limit of two hours.
More recently, the work of Dang et al.~\cite{Dang13b} was extended by El-Hajj et al.~\cite{ElHajj2016}, where the authors attempt to solve
the same formulation proposed by the former work via a \emph{cutting-plane} algorithm. The algorithm explores intermediate models
obtained by considering only a subset of the vehicles and uses the information iteratively retrieved to solve the original
problem. Here, the promising inequalities introduced by Dang et al.~\cite{Dang13b} are also used to accelerate the convergence of
the algorithm.

Overall, the branch-and-price of Keshtkaran et al.~\cite{Keshtkaran16} and the \emph{cutting-plane} algorithm of
El-Hajj et al.~\cite{ElHajj2016} outperform the other exact algorithms previously discussed. In fact, they present a complementary behaviour
when solving the hardest instance sets, i.e., on some instances, one is better than the other and vice-versa.
As pointed out in both works, such behaviour constitutes a pattern between branch-and-cut and branch-and-price
algorithms previously presented in the literature of TOP.

A more recent work of Bianchessi et al.~\cite{Bianchessi2017} introduced a two-index compact (with a polynomial number of variables and
constraints) formulation inspired by the one of Maffioli and Sciomachen~\cite{Maffioli1997} for the \emph{sequential ordering problem},
a scheduling problem where jobs have to be processed on a single machine and are subject to time windows and precedence
relations. In \cite{Bianchessi2017}, the compact formulation for TOP is reinforced with connectivity constraints
and solved via a branch-and-cut algorithm developed with the \emph{callback} mechanism of the optimization solver
CPLEX\footnote{http://www-01.ibm.com/software/commerce/optimization/cplex-optimizer/}. This simple approach showed to be very effective in practice. In fact, 
the algorithm was able to solve at optimality 26 more instances than any other exact algorithm aforementioned
when enabling multi-threading, and 10 more instances when not. All experiments used the CPLEX built-in cuts.

For detailed surveys on exact and heuristic resolution approaches for TOP and its variants, we refer to the works of Vansteenwegen et al.~\cite{Vansteenwegen2011} and Gunawan et al.~\cite{Gunawan2016}.

\section{Problem definition and notation}
\label{s_notation}
STOP is defined on a digraph $G=(N,A)$, where $N$ is the vertex set, and $A$ is the arc set. Let $s,t \in N$ be the origin and the destination vertices, respectively, with $s \neq t$. Moreover, let $S \subseteq N\backslash\{s,t\}$ be the subset of \emph{mandatory}
vertices, and $P \subseteq N\backslash\{s,t\}$ be the set of \emph{profitable} vertices, 
such that $N = S \cup P \cup \{s,t\}$ and $S \cap P = \emptyset$. A reward $p_i \in \mathbb{Z}^+$ is associated with each vertex $i \in P$, and a traverse time $d_{ij} \in \mathbb{R}^+$ is associated with each arc $(i,j) \in A$.
Each vehicle of the homogeneous fleet $M$ can run for no more than a time limit $T$.

STOP aims at finding up to $m = |M|$ routes from $s$ to $t$ such that every mandatory vertex in $S$ belongs to exactly one route and the total sum of rewards collected by visiting profitable vertices is maximized. Here, each profitable vertex in $P$ can be visited by at most one vehicle, thus avoiding the multiple collection of a same reward.
Likewise, each mandatory vertex in $S$ must be visited only once.

In the remainder of this work, we also consider the notation described as follows.
Given a subset $V \subset N$, we define the sets of arcs leaving and entering $V$ as $\delta^+(V) = \{(i,j)\in A:\, i \in V,\, j \in N \backslash V\}$ and $\delta^-(V) = \{(i,j)\in A:\, i \in N\backslash V,\, j \in V\}$, respectively. Similarly, given a vertex $i \in N$,
we define the sets of vertices $\delta^+(i) = \{j\in N:\, (i,j) \in A\}$ and $\delta^-(i) = \{j\in N:\, (j,i)\in A\}$.
Moreover, given two arbitrary vertices $i,j \in N$ and a path $p$ from $i$ to $j$ in $G$, we define $A_p \subseteq A$ as the arc set of $p$. 

Let $R_{ij}$ denote the minimum time needed to reach a vertex $j$ when departing from a vertex $i$ in the graph $G$, i.e., $R_{ij} = \min\{\sum\limits_{a \in A_p}{d_a}:\,\mbox{$p$ is a path from $i$ to $j$ in $G$}\}$. Accordingly, $R_{ii}$ = 0 for all $i \in N$. This $R$ matrix, which is used within the definition of the mathematical formulations described next, is computed \emph{a priori} (for each instance) by means of the classical dynamic programming algorithm of Floyd-Warshall \cite{Cormen2001}. 

\section{Mathematical formulations}
\label{s_models}
In this section, we present two compact Mixed Integer Linear Programming (MILP) formulations for STOP. The first one, denoted by $\mathcal{F}_1$, directly extends the TOP formulation of Bianchessi et al.~\cite{Bianchessi2017} through the addition of constraints that impose the selection of mandatory vertices. The second one, denoted by $\mathcal{F}_2$, is a commodity-based formulation which, to the best of our knowledge, is also introduced in this work. In particular, $\mathcal{F}_1$ and $\mathcal{F}_2$ constitute, respectively, the backbone of the branch-and-cut baseline algorithm (discussed in Section~\ref{s_branch-and-cut}) and of our cutting-plane algorithm (presented in Section~\ref{s_cutting-plane}). By the end of this section, we also give a formal proof of the equivalence of these formulations and discuss how we take advantage of a specific characteristic of $\mathcal{F}_2$ in the cutting-plane algorithm we propose. Moreover, we shortly describe some of the formulations that performed poorly in pilot experiments and were, thus, discarded from this study.

Now, consider the decision variables $y$ on the choice of vertices belonging or not to the solution routes, such that $y_{i} = 1$ if the vertex $i \in N$ is visited by a vehicle
of the fleet, and $y_{i} = 0$, otherwise. Likewise, let the binary variables $x$ identify the solution routes themselves: $x_{ij} = 1$ if the arc $(i,j) \in A$ is traversed in the solution; $x_{ij} = 0$, otherwise. In addition, let the continuous variables $z_{ij}$, for all $(i,j) \in A$, represent the arrival time at vertex $j$ of a vehicle directly coming from vertex $i$. The slack variable $\varphi$ represents the number of vehicles that are not used in the solution.
$\mathcal{F}_1$ is defined from (\ref{b100}) to (\ref{b113}).
\begin{eqnarray}
  \mbox{($\mathcal{F}_1$)\quad}\max && \sum \limits_{i \in P}{p_{i}y_{i}}, \label{b100}\\
  s.t. && y_i = 1 \qquad \forall\, i \in S\cup\{s,t\}, \label{b101} \\
	&& \sum \limits_{j \in \delta^{+}(i)}{x_{ij}} = y_i \qquad \forall\, i \in S\cup P, \label{b102} \\  
	&& \sum \limits_{j \in \delta^{+}(s)}{x_{sj}} = \sum \limits_{i \in \delta^{-}(t)}{x_{it}} = m - \varphi, \label{b103} \\
	&& \sum \limits_{i \in \delta^{-}(s)}{x_{is}} = \sum \limits_{j \in \delta^{+}(t)}{x_{tj}} = 0, \label{b104} \\
	&& \sum \limits_{j \in \delta^{+}(i)}{x_{ij}} - \sum \limits_{j \in \delta^{-}(i)}{x_{ji}} = 0 \qquad \forall\, i \in S\cup P,\label{b105} \\
	&& z_{sj} = d_{sj}x_{sj} \qquad \forall\, j \in \delta^+(s), \label{b106} \\
	&& \sum \limits_{j \in \delta^{+}(i)}{z_{ij}}-\sum \limits_{j \in \delta^{-}(i)}{z_{ji}}  = \sum \limits_{j \in \delta^{+}(i)}{d_{ij}x_{ij}} \qquad \forall\, i \in S\cup P,\label{b107} \\
	&& z_{ij} \leq (T-R_{jt})x_{ij} \qquad \forall (i,j) \in A, \label{b108} \\
	&& z_{ij} \geq (R_{si} + d_{ij})x_{ij} \qquad \forall (i,j) \in A, \label{b109} \\
% 	&& \sum\limits_{(i,j) \in A}d_{ij}x_{ij} \leq mT, \label{b110} \\
    && x_{ij} \in \{0,1\} \qquad \forall (i,j) \in A, \label{b110} \\
    && y_{i} \in \{0,1\} \qquad \forall i\, \in N, \label{b111} \\
    && z_{ij} \geq 0 \qquad \forall (i,j) \in A, \label{b112} \\
    && 0 \leq \varphi \leq m. \label{b113}
\end{eqnarray}

The objective function in (\ref{b100}) gives the total reward collected by visiting profitable vertices. Constraints
(\ref{b101}) impose that all mandatory vertices (as well as $s$ and $t$) are selected, while constraints (\ref{b102}) ensure that each vertex in $S \cup P$ is visited at most once. Restrictions (\ref{b103}) ensure that at most $m$ vehicles
leave the origin $s$ and arrive at the destination $t$, whereas constraints (\ref{b104}) impose that vehicles cannot arrive at
$s$ nor leave $t$. Moreover, constraints (\ref{b105}), along with constraints (\ref{b101}) and (\ref{b102}), guarantee that,
if a vehicle visits a vertex $i \in S \cup P$, then it must enter and leave this vertex exactly once.

Constraints (\ref{b106})-(\ref{b108}) ensure that each of the solution routes from $s$ to $t$ has a total traverse time of at most $T$. In particular, constraints (\ref{b106}) implicitly set the depart time from vertex $s$ to be zero, while constraints (\ref{b107}) manage the subsequent arrival times according to the vertices previously visited in each route. Constraints (\ref{b108}) impose that an arc $(i,j) \in A$ can only be traversed if the minimum extra time needed to reach $t$ from $j$ does not make infeasible the route it belongs. Restrictions (\ref{b109}) are, in fact, valid inequalities
that provide lower bounds on the arrival times represented by the $z$ variables, and
restrictions (\ref{b110})-(\ref{b113}) set the domain of the variables.
Notice that, in $\mathcal{F}_1$, the continuous variables $z$ work as flow variables, thus preventing the
existence of sub-tours in the solutions.

Formulation $\mathcal{F}_1$ was originally stated in \cite{Bianchessi2017} with the additional inequality
\begin{equation}
 \sum\limits_{(i,j)\in A}{d_{ij}x_{ij}} \leq mT, \label{ineq00}
\end{equation}

\noindent which the authors claimed to strengthen the formulation.
By the end of this section (see Corollary~\ref{corol01}), we prove that such inequality is, in fact, redundant.

The commodity-based formulation we introduce in this work, namely $\mathcal{F}_2$, also considers the $y$ and $x$ decision variables as defined above, and the intuition behind it is also similar to that of $\mathcal{F}_1$.
Precisely, in $\mathcal{F}_2$, time is treated as a commodity to be spent by the vehicles when traversing each arc in their routes, such that every vehicle departs from $s$ with an initial amount of $T$ units of commodity, the time limit.
Accordingly, the $z$ variables of $\mathcal{F}_1$ (related to the arrival times at vertices) are replaced, in $\mathcal{F}_2$, by the flow variables $f_{ij}$, for all $(i,j) \in A$, which represent the amount of time still available for a vehicle after traversing the arc $(i,j)$ as not to exceed $T$. As in $\mathcal{F}_1$, the slack variable $\varphi$ represents the number of vehicles that are not used in the solution.
$\mathcal{F}_2$ is defined as follows.
\begin{eqnarray}
  \mbox{($\mathcal{F}_2$)\quad}\max && \sum \limits_{i \in P}{p_{i}y_{i}}, \label{csc100}\\
  s.t. && \mbox{Constraints (\ref{b101})-(\ref{b105})} \nonumber\\
    %&& y_i = 1 \qquad \forall\, i \in S\cup\{s,t\}, \label{csc101} \\
	%&& \sum \limits_{j \in \delta^{+}(i)}{x_{ij}} = y_i \qquad \forall\, i \in S\cup P, \label{csc102} \\  
	%&& \sum \limits_{j \in \delta^{+}(s)}{x_{sj}} = \sum \limits_{i \in \delta^{-}(t)}{x_{it}} = m, \label{csc103} \\
	%&& \sum \limits_{i \in \delta^{-}(s)}{x_{is}} = \sum \limits_{j \in \delta^{+}(t)}{x_{tj}} = 0, \label{csc104} \\
	%&& \sum \limits_{j \in \delta^{+}(i)}{x_{ij}} - \sum \limits_{j \in \delta^{-}(i)}{x_{ji}} = 0 \qquad \forall\, i \in S\cup P,\label{csc105} \\
%	&& \sum \limits_{j \in \delta^{+}(s)}{f_{sj}} =  \sum \limits_{i \in S\cup P}\sum \limits_{j \in \delta^{+}(i)}{d_{ij}x_{ij}},\label{csc106} \\
	&& f_{sj} = (T-d_{sj})x_{sj} \qquad \forall\, j \in \delta^+(s), \label{csc107}\\
	&& \sum \limits_{j \in \delta^{-}(i)}{f_{ji}} - \sum \limits_{j \in \delta^{+}(i)}{f_{ij}} = \sum \limits_{j \in \delta^{+}(i)}{d_{ij}x_{ij}}\qquad \forall\, i \in S\cup P,\label{csc108} \\
	&& f_{ij} \leq (T- R_{si} - d_{ij})x_{ij} \qquad \forall (i,j) \in A,\, i \neq s, \label{csc109} \\
	&& f_{ij} \geq R_{jt}x_{ij} \qquad \forall (i,j) \in A, \label{csc110} \\
    && x_{ij} \in \{0,1\} \qquad \forall (i,j) \in A, \label{csc111} \\
    && y_{i} \in \{0,1\} \qquad \forall i\, \in N, \label{csc112} \\
    && f_{ij} \geq 0 \qquad \forall (i,j) \in A, \label{csc113} \\
    && 0 \leq \varphi \leq m. \label{csc114}
\end{eqnarray}

The objective function in (\ref{csc100}) gives the total reward collected by visiting profitable vertices.
Constraints (\ref{csc107})-(\ref{csc109}) ensure that each of the solution routes has a total traverse time of at most $T$. Precisely,
%constraint (\ref{csc106}) specifies that the flow available at the origin $s$ is exactly the amount consumed along the solution routes.
restrictions (\ref{csc107}) implicitly state, along with (\ref{b103}), that the total flow available at the origin $s$ is $(m-\varphi)T$, and, in particular, each vehicle (used) has an initial amount of $T$
units of flow. Constraints (\ref{csc108}) manage the flow consumption incurred from traversing the arcs selected, whereas constraints (\ref{csc109}) impose that an arc $(i,j) \in A$ can only be traversed if the minimum time of a route from $s$ to $j$ through $(i,j)$ does not exceed $T$. In (\ref{csc109}), we do not consider the arcs leaving the origin, as they are already addressed by (\ref{csc107}).
Similarly to (\ref{b109}), the valid inequalities (\ref{csc110}) give lower bounds on the flow passing through each arc, and constraints (\ref{csc111})-(\ref{csc114}) define the domain of the variables. Here, the management of the flow associated with the variables $f$ also avoids the existence of sub-tours in the solutions.

Notice that, in both formulations, the $y$ variables can be easily discarded, as they solely aggregate specific subsets of the $x$ variables. Nevertheless, they enable us to represent some families of valid inequalities (as detailed in Section~\ref{s_cuts}) by means of less dense cuts, which can noticeably benefit the performance of cutting-plane algorithms.

Now, let $\mathcal{L}_1$ and $\mathcal{L}_2$ be the linearly relaxed versions of $\mathcal{F}_1$ and $\mathcal{F}_2$, respectively.
\begin{theorem}
\label{teo01}
$\mathcal{L}_1$ and $\mathcal{L}_2$ are equivalent.
\end{theorem}
\begin{proof}
 We show that, for every solution in $\mathcal{L}_1$, there is a corresponding one in ${\mathcal{L}_2}$ (and vice-versa), with a same objective function value associated. First, consider $x$ and $y$ as defined in $\mathcal{F}_1$ and $\mathcal{F}_2$, but without the integrality. Also recall that both formulations have the same objective function and that constraints (\ref{b101})-(\ref{b105}) belong to $\mathcal{F}_1$ and to $\mathcal{F}_2$. Then, we only have to show that there also exists a correspondence between the remaining linear constraints 
which define $\mathcal{L}_1$ and $\mathcal{L}_2$. To this end, let us establish the following relation between $z$ and $f$ variables:
\begin{equation}
z_{ij} = Tx_{ij} - f_{ij} \qquad \forall (i,j) \in A.\label{teo01.1}
\end{equation}

From (\ref{teo01.1}), it holds that
\begin{enumerate}
\item $(\ref{b106}) \Longleftrightarrow (\ref{csc107})$
\begin{align*}
z_{sj} = d_{sj}x_{sj} \qquad \forall\, j \in \delta^+(s)
& \;\Longleftrightarrow\; Tx_{sj} - f_{sj} = d_{sj}x_{sj} \qquad \forall\, j \in \delta^+(s)
\\ &\;\Longleftrightarrow\; f_{sj} = (T - d_{sj})x_{sj} \qquad \forall\, j \in \delta^+(s).
\end{align*}
\item $(\ref{b107}) \Longleftrightarrow (\ref{csc108})$
\begin{align*}
\sum \limits_{j \in \delta^{+}(i)}{z_{ij}}-\sum \limits_{j \in \delta^{-}(i)}{z_{ji}}  = \sum \limits_{j \in \delta^{+}(i)}{d_{ij}x_{ij}} \qquad \forall\, i \in S\cup P &\;\Longleftrightarrow \\
\sum \limits_{j \in \delta^{+}(i)}{ \Big(Tx_{ij} - f_{ij} \Big) }-\sum \limits_{j \in \delta^{-}(i)}{ \Big(Tx_{ji} - f_{ji} \Big) }  = \sum \limits_{j \in \delta^{+}(i)}{d_{ij}x_{ij}} \qquad \forall\, i \in S\cup P &\;\Longleftrightarrow \\
\sum \limits_{j \in \delta^{+}(i)}Tx_{ij} - \sum \limits_{j \in \delta^{+}(i)}f_{ij} - \sum \limits_{j \in \delta^{-}(i)}Tx_{ji} + \sum \limits_{j \in \delta^{-}(i)}f_{ji}  = \sum \limits_{j \in \delta^{+}(i)}{d_{ij}x_{ij}} \qquad \forall\, i \in S\cup P &\;\Longleftrightarrow \\
\sum \limits_{j \in \delta^{-}(i)}f_{ji} - \sum \limits_{j \in \delta^{+}(i)}f_{ij} + T\Bigg(\sum \limits_{j \in \delta^{+}(i)}x_{ij} - \sum \limits_{j \in \delta^{-}(i)}x_{ji} \Bigg) = \sum \limits_{j \in \delta^{+}(i)}{d_{ij}x_{ij}} \qquad \forall\, i \in S\cup P,&
\end{align*}

\noindent which, from (\ref{b105}), implies
\begin{align*}
\sum \limits_{j \in \delta^{-}(i)}{f_{ji}} - \sum \limits_{j \in \delta^{+}(i)}{f_{ij}} = \sum \limits_{j \in \delta^{+}(i)}{d_{ij}x_{ij}}\qquad \forall\, i \in S\cup P.
\end{align*}
\item $(\ref{b108}) \Longleftrightarrow (\ref{csc110})$
\begin{align*}
z_{ij} \leq (T-R_{jt})x_{ij} \qquad \forall (i,j) \in A
& \;\Longleftrightarrow\; Tx_{ij} - f_{ij} \leq (T-R_{jt})x_{ij} \qquad \forall (i,j) \in A
\\ &\;\Longleftrightarrow\; - f_{ij} \leq (T-R_{jt})x_{ij} - Tx_{ij} \qquad \forall (i,j) \in A
\\ &\;\Longleftrightarrow\; f_{ij} \geq R_{jt}x_{ij} \qquad \forall (i,j) \in A.
\end{align*}
\item $(\ref{b109}) \;\Longrightarrow\; (\ref{csc109})$
\begin{align*}
z_{ij} \geq (R_{si} + d_{ij})x_{ij} \qquad \forall (i,j) \in A
& \;\Longrightarrow\; Tx_{ij} - f_{ij} \geq (R_{si} + d_{ij})x_{ij} \qquad \forall (i,j) \in A
\\ &\;\Longrightarrow\; f_{ij} \leq (T - R_{si} - d_{ij})x_{ij} \qquad \forall (i,j) \in A
\\ &\;\Longrightarrow\; f_{ij} \leq (T - R_{si} - d_{ij})x_{ij} \qquad \forall (i,j) \in A,\, i \neq s.
\end{align*}

\item $(\ref{csc109}) \mbox{ and } (\ref{csc107}) \;\Longrightarrow\; (\ref{b109})$

From (\ref{csc109}) and (\ref{csc107}), we have that $f_{ij} \leq (T - R_{si} - d_{ij})x_{ij}$ for all $(i,j) \in A$, which implies
\begin{align*}
Tx_{ij} - f_{ij} \geq (R_{si} + d_{ij})x_{ij} \qquad \forall (i,j) \in A
&\;\Longrightarrow\; z_{ij} \geq (R_{si} + d_{ij})x_{ij} \qquad \forall (i,j) \in A. \qedhere
\end{align*}
\end{enumerate}
\end{proof}

\begin{proposition}
\label{prop01}
The inequality (\ref{ineq00}) does not cut off any solution from the polyhedron $\mathcal{L}_2$, the linear relaxation of $\mathcal{F}_2$.
\end{proposition}
\begin{proof}
We prove this result by showing that (\ref{ineq00}) can be implied by linearly combining some of the linear constraints of $\mathcal{F}_2$.
First, by aggregating all the constraints (\ref{csc108}), we obtain
\begin{equation}
\overbrace{\sum\limits_{i \in S\cup P}\Bigg(\sum \limits_{j \in \delta^{-}(i)}{f_{ji}} - \sum \limits_{j \in \delta^{+}(i)}{f_{ij}} \Bigg)}^{(a)} =
\sum\limits_{i \in S\cup P}\sum \limits_{j \in \delta^{+}(i)}{d_{ij}x_{ij}}.\label{ineq0.1}
\end{equation}

Now, let us define the set $\bar{A}=\{(i,j) \in A:\,i,j \in S\cup P\}$ composed of the arcs whose corresponding vertices are neither $s$ nor $t$. Then, we can rewrite (a) as
\begin{align}
\notag \overbrace{\sum\limits_{i \in S\cup P}\Bigg(\sum \limits_{j \in \delta^{-}(i)}{f_{ji}} - \sum \limits_{j \in \delta^{+}(i)}{f_{ij}} \Bigg)}^{(a)}
&= \overbrace{\sum\limits_{i \in S\cup P}\sum \limits_{j \in \delta^{-}(i)}{f_{ji}}}^{(b)} -
\overbrace{\sum\limits_{i \in S\cup P} \sum \limits_{j \in \delta^{+}(i)}{f_{ij}}}^{(c)}
\\ \notag &= \overbrace{\sum\limits_{(j,i) \in \bar{A}}f_{ji} + \sum\limits_{i \in \delta^{+}(s)\cap(S\cup P)}{f_{si}} + \sum\limits_{i \in \delta^{+}(t)\cap(S\cup P)}{f_{ti}}}^{=\,(b)}
\\ & - \overbrace{\Bigg( \sum\limits_{(i,j) \in \bar{A}}f_{ij} + 
\sum\limits_{i \in \delta^{-}(s)\cap(S\cup P)}{f_{is}} + \sum\limits_{i \in \delta^{-}(t)\cap(S\cup P)}{f_{it}} \Bigg)}^{=\,(c)}.\label{ineq0.2}
\end{align}

From (\ref{b104}), (\ref{csc107}) and (\ref{csc109}), we have that $\sum\limits_{i \in \delta^{+}(t)\cap(S\cup P)}{f_{ti}}$ = $\sum\limits_{i \in \delta^{-}(s)\cap(S\cup P)}{f_{is}} = 0$. Also notice that $\sum\limits_{(j,i) \in \bar{A}}f_{ji}$ turns into $\sum\limits_{(i,j) \in \bar{A}}f_{ij}$ (and vice-versa) by simply reordering the notation. Thus, (\ref{ineq0.2}) can be rewritten as
\begin{align}
\overbrace{\sum\limits_{i \in S\cup P}\Bigg(\sum \limits_{j \in \delta^{-}(i)}{f_{ji}} - \sum \limits_{j \in \delta^{+}(i)}{f_{ij}} \Bigg)}^{(a)} = \sum\limits_{i \in \delta^{+}(s)\cap(S\cup P)}{f_{si}} - \sum\limits_{i \in \delta^{-}(t)\cap(S\cup P)}{f_{it}}. \label{ineq0.3}
\end{align}

Directly from (\ref{ineq0.1}) and (\ref{ineq0.3}), it follows that
\begin{align}
\notag \sum \limits_{i \in S\cup P}\sum \limits_{j \in \delta^{+}(i)}{d_{ij}x_{ij}}
  & = \sum\limits_{i \in \delta^{+}(s)\cap(S\cup P)}{f_{si}} - \sum\limits_{i \in \delta^{-}(t)\cap(S\cup P)}{f_{it}}
  \\ \notag & \leq \sum\limits_{i \in \delta^{+}(s)\cap(S\cup P)}{f_{si}}
  \\ & \leq \sum\limits_{i \in \delta^{+}(s)\cap(S\cup P)}{f_{si}} + \sum\limits_{i \in \delta^{+}(s)\backslash(S\cup P)}{f_{si}}
  = \sum\limits_{i \in \delta^{+}(s)}{f_{si}}. \label{ineq0.4}
\end{align}

From (\ref{csc107}) and (\ref{ineq0.4}), we have that
\begin{align}
  \notag \sum \limits_{i \in S\cup P}\sum \limits_{j \in \delta^{+}(i)}{d_{ij}x_{ij}}
  &\leq \sum \limits_{i \in \delta^{+}(s)}{f_{si}} 
  %\\ \notag &\leq \sum \limits_{i \in \delta^{+}(s)}{(T- \overbrace{R_{ss}}^{=\,0} - d_{si})x_{si}}
  \\ \notag &= \sum \limits_{i \in \delta^{+}(s)}{(T- d_{si})x_{si}} 
  \\ &= T\Bigg(\sum\limits_{i \in \delta^{+}(s)}{x_{si}} \Bigg) - \sum\limits_{i \in \delta^{+}(s)}{d_{si}x_{si}},
\end{align}
\noindent which implies
\begin{equation}
 \overbrace{\sum \limits_{i \in S\cup P}\sum \limits_{j \in \delta^{+}(i)}{d_{ij}x_{ij}} + \sum\limits_{i \in \delta^{+}(s)}{d_{si}x_{si}}}^{(d)} \leq T\Bigg(\sum\limits_{i \in \delta^{+}(s)}{x_{si}} \Bigg). \label{ineq01}
\end{equation}

Notice that (d) corresponds to 
\begin{align}
\notag \overbrace{\sum \limits_{i \in S\cup P}\sum \limits_{j \in \delta^{+}(i)}{d_{ij}x_{ij}} + \sum\limits_{i \in \delta^{+}(s)}{d_{si}x_{si}}}^{(d)}
& = \sum \limits_{i \in S\cup P\cup\{s\}}\sum \limits_{j \in \delta^{+}(i)}{d_{ij}x_{ij}}
\\ & = \sum \limits_{i \in N \backslash \{t\}}\sum \limits_{j \in \delta^{+}(i)}{d_{ij}x_{ij}}
= \sum \limits_{(i,j) \in A}{d_{ij}x_{ij}}, \label{ineq02}
\end{align}
\noindent since, from (\ref{b104}), no arc leaving $t$ can be selected. Then, from (\ref{b103}), (\ref{csc114}), (\ref{ineq01}) and (\ref{ineq02}), it follows that 
\begin{align*}
\overbrace{\sum \limits_{(i,j) \in A}{d_{ij}x_{ij}}}^{=\,(d)}
& \leq T\Bigg(\sum\limits_{j \in \delta^{+}(s)}{x_{sj}} \Bigg)
 = (m-\varphi)T \leq mT. \qedhere
\end{align*}
\end{proof}

\begin{corollary}
\label{corol01}
The inequality (\ref{ineq00}) does not cut off any solution from the polyhedron
$\mathcal{L}_1$, the linear relaxation of $\mathcal{F}_1$.
\end{corollary}
\begin{proof}
Directly from Theorem~\ref{teo01} and Proposition~\ref{prop01}.
\end{proof}

As already mentioned, constraints (\ref{b109}) and (\ref{csc110}) are, in fact, valid inequalities in formulations $\mathcal{F}_1$ and $\mathcal{F}_2$, respectively.
Accordingly, one can take advantage of this fact when solving these formulations by means of branch-and-cut schemes that, like CPLEX, have cut management mechanisms. Precisely, in such schemes, valid inequalities --- usually referred to as \emph{user cuts} --- are treated differently from actual model restrictions, as they are stored in \emph{pools} of cuts and only added to the models whenever they are violated.
%The aim of this mechanism is to preserve the bound gain incurred by the addition of the valid inequalities while generating possibly lighter models within the branch-and-bound tree.

In the case of $\mathcal{F}_2$, one can particularly benefit from cut management mechanisms.
Precisely, we experimentally observed that, on average, the impact of the valid inequalities (\ref{csc110}) on the strength of $\mathcal{F}_2$ is not as expressive as that of inequalities (\ref{b109}) on the strength of $\mathcal{F}_1$ (we refer to \ref{appendix_0} for the summary of the results). In practice, this behaviour suggests that (\ref{csc110}) are less likely to be active at optimal solutions for $\mathcal{F}_2$. Then, instead of treating inequalities (\ref{csc110}) as constraints of $\mathcal{F}_2$, we can explicitly define them as user cuts as an attempt to make the corresponding models lighter while not losing the strength of the original formulation. This simple idea was originally proposed by Fischetti et al.~\cite{Fischetti98} and is our main motivation for solving $\mathcal{F}_2$ --- and not $\mathcal{F}_1$ --- within our cutting-plane algorithm.

Originally, we also proposed and tested several other compact formulations for STOP. Precisely, we tested two-commodity and multi-commodity versions of the formulations $\mathcal{F}_1$ and $\mathcal{F}_2$, as well as variations in which the $x$ variables are indexed by the vehicles of the fleet $M$. Likewise, we also considered classical commodity-based formulations in which the consumption of each unit of commodity is linked to visiting a single vertex. In addition, we adapted to STOP the formulation for CVRP proposed by Kulkarni and Bhave~\cite{KULKARNI1985}, which uses the reinforced Miller-Tucker-Zemlin~\cite{Miller1960} subtour elimination constraints of Kara et al.~\cite{KARA2004}. Nevertheless, since all of these additional formulations performed poorly (when solved directly with CPLEX) in comparison with $\mathcal{F}_1$ and $\mathcal{F}_2$, they were omitted from this study. In fact, we conjecture that the superiority of $\mathcal{F}_1$ and $\mathcal{F}_2$ is partially due to the way they implicitly handle the limit imposed on the total traverse times of the routes.

\section{Families of valid inequalities}
\label{s_cuts}
In this section, we discuss three families of valid inequalities to be separated in the cutting-plane scheme we propose. 
They consist of some general connectivity constraints, a class of conflict cuts and classical lifted cover inequalities based on dual bounds. %In particular, these cover inequalities work as \emph{logic cuts} (see, e.g., \cite{Fischetti98,Hooker94a,Hooker94b}), as they might cut off feasible integer solutions. Moreover, as far as we are aware, the second class of inequalities is also introduced in this work.

\subsection{General Connectivity Constraints (GCCs)}
The GCC inequalities were originally devised to ensure connectivity and prevent sub-tours in solution routes \cite{Toth2001}. Although these properties are already guaranteed in formulations $\mathcal{F}_1$ and $\mathcal{F}_2$, the GCCs presented below are able to further strengthen both formulations \citep{Bianchessi2017}.
\begin{equation}
\sum\limits_{(i,j) \in \delta^+(V)}{x_{ij}} \geq y_{k} \qquad \forall\; V \subseteq N\backslash\{t\},\, |V| \geq 2,\; \forall \, k \in V.
\end{equation}

\subsection{Conflict Cuts (CCs)}
Consider the set $\mathcal{K}$ of vertex pairs which cannot be simultaneously in a same valid route. Precisely, for every pair $\langle i,j\rangle \in \mathcal{K}$, with $i,j \in N\backslash\{s,t\}$, we have that any route from $s$ to $t$ that visits $i$ and $j$ (in any order) has a total traverse time that exceeds the limit $T$. Then, CCs are defined as follows.
\begin{eqnarray}
\sum\limits_{e \in \delta^-(V)}{x_{e}} \geq y_{i} + y_j \qquad \forall\; \langle i,j \rangle \in \mathcal{K},\; \forall\; V \subseteq N\backslash\{s\},\, \{i,j\} \subseteq V, \label{cc00}\\
\sum\limits_{e \in \delta^+(V)}{x_{e}} \geq y_{i} + y_j \qquad \forall\; \langle i,j \rangle \in \mathcal{K},\; \forall\; V \subseteq N\backslash\{t\},\, \{i,j\} \subseteq V.\label{cc01}
\end{eqnarray}

\begin{proposition}
\label{prop02}
Inequalities (\ref{cc00}) do not cut off any feasible solution of $\mathcal{F}_2$.
\end{proposition}
\begin{proof}
Consider an arbitrary feasible solution $(\bar{x},\bar{y},\bar{f},\bar{\varphi})$ for $\mathcal{F}_2$, a pair of conflicting vertices $\langle i,j \rangle \in \mathcal{K}$ and a subset $V \subseteq N \backslash \{s\}$. Then, we have four possibilities:
\begin{enumerate}
\item if ${\bar y}_i = {\bar y}_j = 0$, then $\overbrace{\sum\limits_{e \in \delta^-(V)}{x_{e}}}^{\geq\,0} \geq \overbrace{y_{i} + y_j}^{=\,0}$.
\item if ${\bar y}_i = 1$ and ${\bar y}_j = 0$, then, since $s \notin V$, there must be an arc $e' \in \delta^-(V)$, with
$\bar{x}_{e'} = 1$, so that the vertex $i$ is traversed in a route from $s$. Thus, $\overbrace{\sum\limits_{e \in \delta^-(V)}{x_{e}}}^{\geq\,1} \geq \overbrace{y_{i} + y_j}^{=\,1}$.
\item if ${\bar y}_i = 0$ and ${\bar y}_j = 1$, then, since $s \notin V$, there must be an arc $e' \in \delta^-(V)$, with
$\bar{x}_{e'} = 1$, so that the vertex $j$ is traversed in a route from $s$. Thus, $\overbrace{\sum\limits_{e \in \delta^-(V)}{x_{e}}}^{\geq\,1} \geq \overbrace{y_{i} + y_j}^{=\,1}$.
\item if ${\bar y}_i = {\bar y}_j = 1$, then, since $i$ and $j$ are conflicting, there must be at least two disjoint routes from $s$ to $t$, one that visits $i$, and other that visits $j$. Since $s \notin V$, we must also have at least two arcs
$e', e'' \in \delta^-(V)$, such that $x_{e'} = x_{e''} = 1$. Therefore, $\overbrace{\sum\limits_{e \in \delta^-(V)}{x_{e}}}^{\geq\,2} \geq \overbrace{y_{i} + y_j}^{=\,2}$.
\end{enumerate}

\end{proof}
\begin{corollary}
\label{corol02}
Inequalities (\ref{cc01}) do not cut off any feasible solution of $\mathcal{F}_2$.
\end{corollary}
\begin{proof}
The same mathematical argumentation of Proposition~\ref{prop02} can be used to prove the validity of inequalities (\ref{cc01}) by simply replacing $s$ and $\delta^-(V)$ with $t$ and $\delta^+(V)$, respectively.  
\end{proof}

\begin{corollary}
\label{corol02.2}
Inequalities (\ref{cc00}) and (\ref{cc01}) do not cut off any feasible solution of $\mathcal{F}_1$.
\end{corollary}
\begin{proof}
Directly from Theorem~\ref{teo01}, Proposition~\ref{prop02} and Corollary~\ref{corol02}.
\end{proof}

Notice that, from (\ref{b105}), we have that $\sum\limits_{e \in \delta^-(V)}{x_{e}} = \sum\limits_{e \in \delta^+(V)}{x_{e}}$ for all $V \subseteq S \cup P = N\backslash\{s,t\}$. Therefore, for all $V \subseteq N\backslash\{s,t\}$, inequalities (\ref{cc00}) and (\ref{cc01}) cut off the exact same regions of the polyhedrons $\mathcal{L}_1$ and $\mathcal{L}_2$. In this sense, the whole set of CCs can be represented in a more compact manner
by (\ref{cc00}) and 
\begin{eqnarray}
\sum\limits_{e \in \delta^+(V)}{x_{e}} \geq y_{i} + y_j \qquad \forall\; \langle i,j \rangle \in \mathcal{K},\; \forall\; V \subseteq N \backslash \{t\},\, \{s,i,j\} \subseteq V.\label{cc02}
\end{eqnarray}

\noindent Notice that inequalities (\ref{cc02}) only consider the subsets $V \subseteq N \backslash \{t\}$ which necessarily contain $s$.

Intuitively speaking, CCs forbid the simultaneous selection (in a same route) of any pair of conflicting vertices of $\mathcal{K}$. In particular, CCs (\ref{cc00}) work similarly to the classical \emph{capacity-cut constraints} of Toth and Vigo~\cite{Toth2001}, but in a more flexible manner. Precisely, given a pair of conflicting vertices $\langle i,j \rangle \in \mathcal{K}$ and a subset $V \subseteq N\backslash\{s\},\, \{i,j\} \subseteq V$, the corresponding CC of type (\ref{cc00}) states that the minimum number of vehicles needed to visit $i$ and $j$ is exactly $y_i + y_j$ (which assumes, at most, value 2). Alternatively, $y_i + y_j$ can be seen as a lower bound on the number of vehicles needed to visit all the vertices in $V$.

For instance, consider the digraph shown in Figure~\ref{fig:ex}, and let the traverse times of its arcs be $d_{si} = d_{sl} = d_{ik} = d_{lj} = d_{lt} = d_{jt} = 1$ and $d_{kj} = 2$. Also consider $S = \emptyset$ and a single vehicle to move from $s$ to $t$, with $T= 4$. In this case, $i$ and $j$ are an example of conflicting vertices, since the only possible route from $s$ to $t$ visiting both vertices exceeds the time limit 4. Figures~\ref{fig:cc00} and~\ref{fig:cc02} show typical fractional solutions that are cut off by CCs $(\ref{cc00})$ and $(\ref{cc02})$, respectively.
Notice that the solution of Figure~\ref{fig:cc00a} also violates CCs $(\ref{cc01})$.
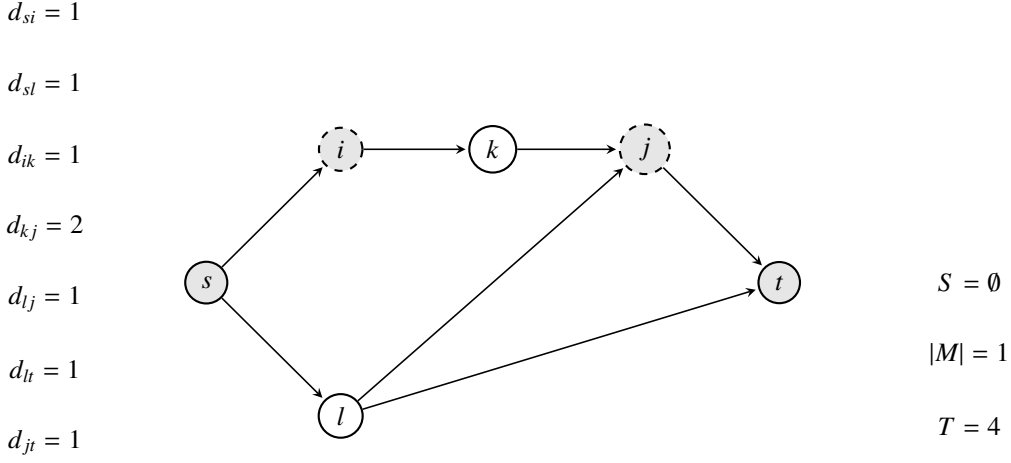
\begin{figure}
\centering
  \begin{tikzpicture}[
            > = stealth, % arrow head style
            shorten > = 1pt, % don't touch arrow head to node
            auto,
            node distance = 2.5cm, % distance between nodes
            semithick % line style
        ]

        \tikzstyle{every state}=[
            draw = black,
            thick,
            fill = white,
            minimum size = 4mm
        ]

        \node[state] (s) [fill = gray!20]{$s$};
        \node[state] (i) [above right of=s,fill = gray!20,dashed] {$i$};
        \node[state] (l) [below right of=s] {$l$};
        \node[state] (k) [right of=i,node distance = 2cm] {$k$};
        \node[state] (j) [right of=k,node distance = 2cm,fill = gray!20,dashed] {$j$};
        \node[state] (t) [below right of=j,fill = gray!20] {$t$};
        
        \node[state] (jt) [below left of = s, draw = none, node distance = 3cm]{$d_{jt}= 1$};
        \node[state] (lt) [above of = jt, draw = none, node distance = 0.95cm]{$d_{lt}= 1$};
        \node[state] (lj) [above of = lt, draw = none, node distance = 0.95cm]{$d_{lj}= 1$};
        \node[state] (kj) [above of = lj, draw = none, node distance = 0.95cm]{$d_{kj}= 2$};
        \node[state] (ik) [above of = kj, draw = none, node distance = 0.95cm]{$d_{ik}= 1$};
        \node[state] (sl) [above of = ik, draw = none, node distance = 0.95cm]{$d_{sl}= 1$};
        \node[state] (si) [above of = sl, draw = none, node distance = 0.95cm]{$d_{si}= 1$};
        
        \node[state] (S) [right of = t, draw = none, node distance = 2.5cm]{$S = \emptyset$};
        \node[state] (M) [below of = S, draw = none, node distance = 0.95cm]{$|M| = 1$};
         \node[state] (T) [below of = M, draw = none, node distance = 0.95cm]{$T = 4$};
    
        \path[->] (s) edge node {} (i);
        \path[->] (s) edge node {} (l);
        \path[->] (l) edge node {} (j);
        \path[->] (i) edge node {} (k);
        \path[->] (k) edge node {} (j);
        \path[->] (j) edge node {} (t);
        \path[->] (l) edge node {} (t);
    
    \end{tikzpicture}
\caption{An example of an STOP instance. Profit values are omitted. Here, the pair $\langle i,j\rangle$ gives an example of conflicting vertices.}
\label{fig:ex}
\end{figure}

\begin{figure}
\centering
\begin{subfigure}{.5\textwidth}
  \centering
  \begin{tikzpicture}[
            > = stealth, % arrow head style
            shorten > = 1pt, % don't touch arrow head to node
            auto,
            node distance = 2.5cm, % distance between nodes
            semithick % line style
        ]

        \tikzstyle{every state}=[
            draw = black,
            thick,
            fill = white,
            minimum size = 4mm
        ]

        \node[state] (s) [fill = gray!20]{$s$};
        \node[state] (i) [above right of=s,dashed,fill = gray!20] {$i$};
        %\node[state] (V) [left of = i, draw = none, node distance = 2cm]{$V$};
        \node[state] (l) [below right of=s] {$l$};
        \node[state] (k) [right of=i,node distance = 2cm] {$k$};
        \node[state] (j) [right of=k,node distance = 2cm,dashed,fill = gray!20] {$j$};
        \node[state] (t) [below right of=j,fill = gray!20] {$t$};
        %\node[state] (V') [left of = l, draw = none, node distance = 2cm]{$N\backslash V$};
        
        \node[state] (V') [below left of = s, draw = none, node distance = 1cm]{$N\backslash V$};
        \node[state] (V) [above of = V', draw = none, node distance = 2cm]{$V$};

        \path[->] (s) edge node {0.3} (i);
        \path[->] (s) edge node {0.7} (l);
        \path[->] (l) edge node[below] {0.7} (t);
        \path[->] (i) edge node {0.3} (k);
        \path[->] (k) edge node {0.3} (j);
        \path[->] (j) edge node {0.3} (t);
        \path[->] (l) edge node {0} (j);
        
        %\draw[green!0.0,fill=green, fill opacity=0.2] (1cm,1cm) rectangle (5.5cm,4cm);
        \draw[red, dashed,thick] (-1.2, 0.55) -- (8, 0.8);
    \end{tikzpicture}
    \caption{Representation of a fractional solution that is cut off from $\mathcal{L}_2$ by (\ref{cc00}). Here, we have $x_{si} = x_{ik} = x_{kj} = x_{jt} = 0.3$, $x_{sl} = x_{lt} = 0.7$ and $x_{lj} = 0$. Accordingly, $y_s = y_t = 1$, $y_i = y_k = y_j = 0.3$ and $y_l = 0.7$. The violated inequality has $V = \{i,k,j\}$.}
  \label{fig:cc00a}
\end{subfigure}
\\%
\begin{subfigure}{.5\textwidth}
  \centering
  \begin{tikzpicture}[
            > = stealth, % arrow head style
            shorten > = 1pt, % don't touch arrow head to node
            auto,
            node distance = 2.5cm, % distance between nodes
            semithick % line style
        ]

        \tikzstyle{every state}=[
            draw = black,
            thick,
            fill = white,
            minimum size = 4mm
        ]

        \node[state] (s) [fill = gray!20]{$s$};
        \node[state] (i) [above right of=s,fill = gray!20,dashed] {$i$};
        \node[state] (l) [below right of=s] {$l$};
        \node[state] (k) [right of=i,node distance = 2cm] {$k$};
        \node[state] (V) [above of = i, draw = none, node distance = 1cm]{$V$};
        \node[state] (V') [left of = V, draw = none, node distance = 1.7cm]{$N\backslash V$};
        \node[state] (j) [right of=k,node distance = 2cm,fill = gray!20,dashed] {$j$};
        \node[state] (t) [below right of=j,fill = gray!20] {$t$};
    
        \path[->] (s) edge node {0.3} (i);
        \path[->] (s) edge node {0.7} (l);
        \path[->] (l) edge node {0.7} (j);
        \path[->] (i) edge node {0.3} (k);
        \path[->] (k) edge node {0.3} (j);
        \path[->] (j) edge node {1} (t);
        \path[->] (l) edge node[below] {0} (t);

        \draw[red, dashed, thick] (0.9, 3.2) -- (0.9,-2);
    \end{tikzpicture}
  \caption{Representation of a fractional solution that is cut off from $\mathcal{L}_2$ by (\ref{cc00}). Here, we have $x_{si} = x_{ik} = x_{kj} = 0.3$, $x_{sl} = x_{lj} = 0.7$, $x_{lt} = 0$ and $x_{jt} = 1$. Accordingly, $y_s = y_t = y_j = 1$, $y_i = y_k = 0.3$ and $y_l = 0.7$. The violated inequality has $V = \{i,k,j,l,t\}$.}
  \label{fig:cc00b}
\end{subfigure}
\caption{Examples of fractional solutions that are cut off by CCs $(\ref{cc00})$ when considering the polyhedron
$\mathcal{L}_2$ and the STOP instance of Figure~\ref{fig:ex}.}
\label{fig:cc00}
\end{figure}
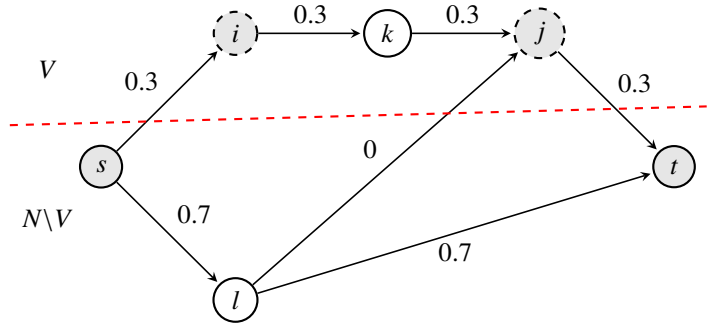
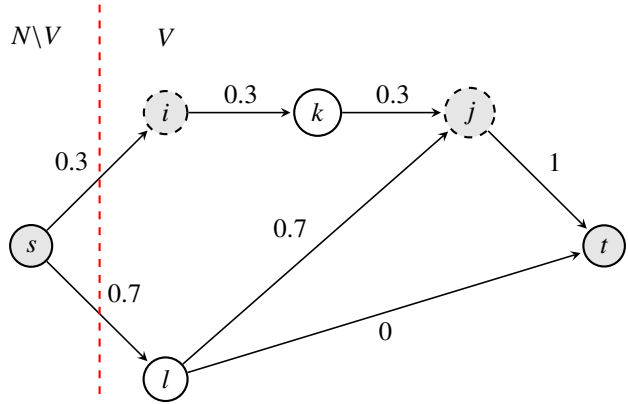

\begin{figure}
\centering
  \begin{tikzpicture}[
            > = stealth, % arrow head style
            shorten > = 1pt, % don't touch arrow head to node
            auto,
            node distance = 2.5cm, % distance between nodes
            semithick % line style
        ]

        \tikzstyle{every state}=[
            draw = black,
            thick,
            fill = white,
            minimum size = 4mm
        ]

        \node[state] (s) [fill = gray!20]{$s$};
        \node[state] (i) [above right of=s,fill = gray!20,dashed] {$i$};
        \node[state] (l) [below right of=s] {$l$};
        \node[state] (k) [right of=i,node distance = 2cm] {$k$};
        \node[state] (j) [right of=k,node distance = 2cm,fill = gray!20,dashed] {$j$};
        \node[state] (t) [below right of=j,fill = gray!20] {$t$};
        
        \node[state] (V) [below left of = l, draw = none, node distance = 1cm]{$V$};
        \node[state] (V') [right of = V, draw = none, node distance = 1.8cm]{$N\backslash V$};
    
        \path[->] (s) edge node {0.5} (i);
        \path[->] (s) edge node {0.5} (l);
        \path[->] (l) edge node {0.5} (j);
        \path[->] (i) edge node {0.5} (k);
        \path[->] (k) edge node {0.5} (j);
        \path[->] (j) edge node {1} (t);
        \path[->] (l) edge node[below] {0} (t);

        \draw[red, dashed, thick] (1.0, -3) -- (7.3,2);
    \end{tikzpicture}
\caption{Example of a fractional solution that is cut off by CCs $(\ref{cc02})$ when considering the polyhedron
$\mathcal{L}_2$ and the STOP instance of Figure~\ref{fig:ex}. Here, we have $x_{si} = x_{ik} = x_{kj} = x_{sl} = x_{lj} = 0.5$, $x_{lt} = 0$ and $x_{jt} = 1$. Accordingly, $y_s = y_t = y_j = 1$ and $y_i = y_k = y_l = 0.5$. The violated inequality has $V = \{s,l,i,k,j\}$.}
\label{fig:cc02}
\end{figure}
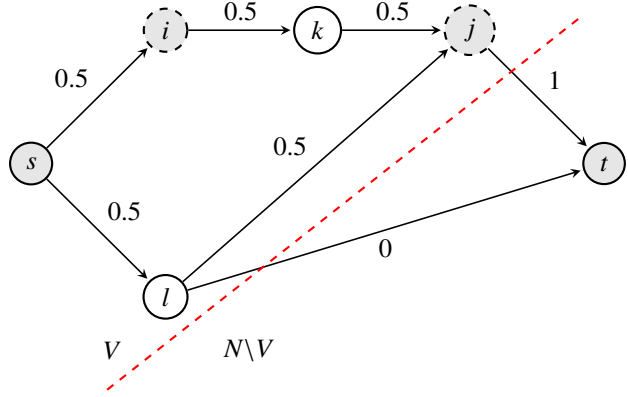

\subsection{Lifted Cover Inequalities (LCIs)}
Let $\tau$ be a dual (upper) bound on the optimal solution value of $\mathcal{F}_2$. Then, consider the inequality
\begin{equation}
\sum \limits_{i \in P}{p_{i}y_{i}} \leq \lfloor\tau\rfloor \label{lci1},
\end{equation}

\noindent where $P \subseteq N \backslash \{s,t\}$ is the set of profitable vertices, with $p_i \in \mathbb{Z}^+$ for all $i \in P$, as defined in Section~\ref{s_notation}. By definition, (\ref{lci1}) is valid for $\mathcal{F}_2$ ( and $\mathcal{F}_1$), once its left-hand side corresponds to the objective function of this formulation.
%Clearly, (\ref{lci1}) is not valid for $\mathcal{F}_2$, as it might cut off feasible solutions. However, since it is satisfied by any optimal solution for $\mathcal{F}_2$, it can still be used as a logic cut~\cite{Hooker94a}. %Here, we assume that $p_i \leq \tau$ for all $i \in P$, since the opposite would imply $y_i = 0$.

Notice that (\ref{lci1}) consists of a \emph{knapsack} constraint, and, thus, it
can be strengthened by means of classical cover inequalities and \emph{lifting}. In this sense, a set $C \subseteq P$ is called a \emph{cover} for inequality (\ref{lci1}) if $\sum \limits_{i \in C}{p_{i}} > \lfloor\tau\rfloor$. Moreover, this cover is said to be \emph{minimal} if it no longer covers (\ref{lci1}) once any of its elements is removed, i.e., $\sum \limits_{i \in C\backslash\{j\}}{p_{i}} \leq \lfloor\tau\rfloor$ for all $j \in C$.

Consider the set $\Phi$ of the $y$ solution values that satisfy (\ref{lci1}). Precisely,
\begin{equation}
\Phi = \Big\{y \in \{0,1\}^{|P|}:\, \sum \limits_{i \in P}{p_{i}y_{i}} \leq \lfloor\tau\rfloor\Big\}. \label{lci1_2}
\end{equation}
\noindent For any cover $C \subseteq P$, the inequality 
\begin{equation}
\sum \limits_{i \in C}{y_{i}} \leq |C| -1 \label{lci2}
\end{equation}

\noindent is called a \emph{cover inequality} and is valid for $\Phi$. In this work, we apply a variation of these inequalities, namely LCIs, which are facet-inducing for $\Phi$ and can be devised from (\ref{lci2}) through lifting (see, e.g., \cite{Balas1975,Wolsey1975,Gu98,Kaparis08}). Let two disjoint sets $C_1$ and $C_2$ define a partition of a given minimal cover $C$, with $C_1 \neq \emptyset$. LCIs are defined as
\begin{equation}
\sum \limits_{i \in C_1}{y_{i}} + \sum\limits_{j \in C_2}{\pi_jy_j} + \sum\limits_{j \in P\backslash C}{\mu_jy_j} \leq |C_1| + \sum\limits_{j \in C_2}{\pi_j} -1, \label{lci3}
\end{equation}

\noindent where $\pi_i \in \mathbb{Z}$, $\pi_i \geq 1$ for all $i\in C_2$, and $\mu_i \in \mathbb{Z}$, $\mu_i \geq 0$ for all $ i\in P \backslash C$, are the lifted coefficients. We detail the way these coefficients are computed in the separation procedure of Section~\ref{s_sep_lcis}. For now, only notice that setting $\pi = \textbf{1}$ and $\mu = \textbf{0}$ suffices for the validity of (\ref{lci3}), as it leads to the classical cover inequality (\ref{lci2}). 

\section{Separation of the valid inequalities}
\label{s_separation}
Let $(\bar{x},\bar{y},\bar{f},\bar{\varphi})$ (or $(\bar{x},\bar{y},\bar{z},\bar{\varphi})$) be a given fractional solution referring to the linear relaxation of $\mathcal{F}_2$ (or $\mathcal{F}_1$). Also consider the residual graph $\tilde{G} = (N,\tilde{A})$ induced by $(\bar{x},\bar{y},\bar{f},\bar{\varphi})$, such that each arc $(i,j) \in A$ belongs to $\tilde{A}$ if, and only if, $\bar{x}_{ij} > 0$. Moreover, a capacity $c[i,j] = \bar{x}_{ij}$ is associated with each arc $(i,j) \in \tilde{A}$.
%The original traverse times associated with the arcs $A$ remain the same in $\tilde{G}$.

As detailed in the sequel, the separation of GCCs and CCs involves solving maximum flow problems. Then, for clarity, consider the following notation.
Given an arbitrary digraph ${G}_a$ with capacitated arcs, and two vertices $i$ and $j$ of ${G}_a$, let ${\textit{max-flow}}_{i\to j}{({G}_a)}$ denote the problem of finding the maximum flow (and, thus, a minimum cut) from $i$ to $j$ on ${G}_a$.
Moreover, let $\langle F_{{i\to j}},\theta_{i\to j} \rangle$ denote an optimal solution of such problem, where $F_{{i\to j}}$ is the value of the maximum flow, and $\theta_{i\to j}$ defines a corresponding minimum cut, with $i \in \theta_{i\to j}$. 

\subsection{GCCs}
\label{s_separate_gccs}
The separation of violated GCCs is done by means of the same algorithm adopted by Bianchessi et al.~\cite{Bianchessi2017} and described as follows.
For each pair of vertices $\langle v,t \rangle$, $v \in N\backslash\{t\}$, a
maximum flow from $v$ to $t$ on $\tilde{G}$ is computed, i.e., the problem
${\textit{max-flow}}_{v\to t}{(\tilde{G})}$ is solved, obtaining a solution $\langle F_{v\to t},\theta_{v\to t} \rangle$.
Moreover, consider $v^* = \arg\max\limits_{j \in \theta_{v\to t}}\{{\bar y}_j\}$.
If $t \notin \theta_{v\to t}$, $|\theta_{v\to t}| \geq 2$ and ${\bar y}_{v^*}$ is greater than the value of the maximum flow $F_{v\to t}$, then a violated GCC is found, which corresponds to 
\begin{equation}
\sum\limits_{(i,j) \in \delta^+(\theta_{v\to t})}{x_{ij}} \geq y_{v^*}.
\end{equation}

\noindent Notice that, in this case, the GCCs found (if any) are the most violated ones, one for each pair $\langle v,t \rangle$, $v \in N\backslash\{t\}$.

\subsection{CCs}
First, we determine the set $\mathcal{K}$ of conflicting vertices. To this end, let $\mathcal{M}$ be a $|N| \times |N|$ matrix such that, for each pair of vertices $i, j \in N$,
$\mathcal{M}[i,j]$ denotes the traverse time of a shortest (minimum time) path from $i$ to $j$ on ${G}$. If no path exists from a vertex to another, the corresponding entry of $\mathcal{M}$
is set to infinity. One may observe that $\mathcal{M}$ is not necessarily symmetric, since ${G}$ is directed. Moreover, considering that the traverse times associated with the arcs of ${G}$ are non-negative (and, thus, no negative cycle exists), one can determine $\mathcal{M}$ by means of the classical dynamic programming algorithm of Floyd-Warshall \cite{Cormen2001}, for instance.

Then, $\mathcal{K}$ can be computed by checking, for all pairs $\langle i,j\rangle$, $i,j \in N$, $i \neq j$, if there exists a path from  $s$ to $t$ on ${G}$ that traverses both $i$ and $j$ (in any order) and that satisfies the total time limit $T$. If no such path exists, then $\langle i,j\rangle$ belongs to $\mathcal{K}$. For simplicity, in this work, we only consider a subset $\tilde{\mathcal{K}} \subseteq \mathcal{K}$ of conflicting vertex pairs, such that
\begin{equation*}
\langle i,j\rangle\in \tilde{\mathcal{K}} \text{ iff }
\begin{cases}
        \quad (i)\; \mathcal{M}[s,i] + \mathcal{M}[i,j] + \mathcal{M}[j,t] > T,\; \text{and} \\
        \quad (ii)\; \mathcal{M}[s,j] + \mathcal{M}[j,i] + \mathcal{M}[i,t] > T, 
\end{cases}
\qquad{ \forall\, i,j \in N, i \neq j.}
\end{equation*}

where $(i)$ is satisfied if a minimum traverse time route from $s$ to $t$ that visits $i$ before $j$ exceeds the time limit. Likewise, $(ii)$ considers a minimum time route that visits $j$ before $i$.
Since the routes from $s$ to $t$ considered in $(i)$ and $(ii)$ are composed by simply aggregating entries of $\mathcal{M}$, they may not be elementary, i.e., they might visit a same vertex more than once. Then,
$\tilde{\mathcal{K}}$ is not necessarily equal to ${\mathcal{K}}$. Also observe that we only have to compute $\tilde{\mathcal{K}}$ a single time for a given STOP instance, as it is completely based on the original graph $G$.

Once $\tilde{\mathcal{K}}$ is computed, we look for violated CCs of types (\ref{cc00}) and (\ref{cc01}) separately, as described in the algorithm of Figure~\ref{cc_separation}. Let set $\mathcal{X}$ keep the CCs found during the separation procedure. Initially, $\mathcal{X}$ is empty (line 1, Figure~\ref{cc_separation}). Then, for all pairs of conflicting vertices $\langle i,j\rangle \in \tilde{\mathcal{K}}$, we build two auxiliary graphs, one for each type of CC. The first graph, denoted by
$\tilde{G}_1$, is built by adding to the residual graph $\tilde{G}$ an artificial vertex $\alpha$ and two arcs: $(i,\alpha)$ and $(j,\alpha)$ (see line 4, Figure~\ref{cc_separation}). The second one, denoted by $\tilde{G}_2$, is built by reversing all the arcs of $\tilde{G}$ and, then, adding an artificial vertex $\beta$, as well as the arcs $(i,\beta)$ and $(j,\beta)$ (see line 5, Figure~\ref{cc_separation}).

The capacities of the arcs of
$\tilde{G}_1$ and $\tilde{G}_2$ are kept in the data structures $c_1$ and $c_2$, respectively. In both graphs,
the capacities of the original arcs in $\tilde{G}$ are preserved (see lines 6-8, Figure~\ref{cc_separation}).
Moreover, all the additional arcs have a same capacity value, which is equal to a sufficiently large number. Here, we adopted the value of $|M|$, the number of vehicles (see line 9, Figure~\ref{cc_separation}).

Figure~\ref{fig:ex_ccs} illustrates the construction of the auxiliary graphs described above. In this example, we consider the STOP instance of Figure~\ref{fig:ex} and assume that the current fractional solution $(\bar{x},\bar{y},\bar{f},\bar{\varphi})$ has ${\bar x}_{si} = {\bar x}_{ik} = {\bar x}_{kj} = {\bar x}_{sl} = {\bar x}_{lj} = 0.5$, ${\bar x}_{lt} = 0$ and ${\bar x}_{jt} = 1$.
\begin{figure}[!ht]
\begin{center}
\scalebox{1}
{
\framebox
{
\begin{minipage}[t]{25cm}
{\small
\begin{tabbing}
xxx\=xxx\=xxx\=xxx\=xxx\=xxx\=xxx\=xxx\=xxx\=xxx\= \kill
\textbf{Input: }A fractional solution $(\bar{x},\bar{y},\bar{f},\bar{\varphi})$, its corresponding residual graph $\tilde{G} = (N,\tilde{A})$ and \\ the subset $\tilde{\mathcal{K}}$ of conflicting vertex pairs. \\
\textbf{Output: }{A set $\mathcal{X}$ of CCs violated by $(\bar{x},\bar{y},\bar{f},\bar{\varphi})$.} \\
\textbf{\scriptsize1.} $\mathcal{X} \leftarrow \emptyset$;\\
\textbf{\scriptsize2.} \textbf{for all }{$\langle i,j \rangle \in \tilde{\mathcal{K}}$} \textbf{ do}\\
\textbf{\scriptsize3.} \> \> \textit{Step \rom{1}. Building the auxiliary graphs}\\
\textbf{\scriptsize4.} \> \>  Build $\tilde{G}_1 = (\tilde{N}_1,\tilde{A}_1)$, with $\tilde{N}_1 = N \cup \{\alpha\}$ and $\tilde{A}_1 = \tilde{A} \cup \{(i,\alpha),(j,\alpha)\}$;\\
\textbf{\scriptsize5.} \> \>  Build $\tilde{G}_2 = (\tilde{N}_2,\tilde{A}_2)$, with $\tilde{N}_2 = N \cup \{\beta\}$ and $\tilde{A}_2 = \{(v,u):\, (u,v) \in \tilde{A}\} \cup \{(i,\beta),(j,\beta)\}$;\\
\textbf{\scriptsize6.} \> \> \textbf{for all }{$(u,v) \in \tilde{A}$} \textbf{ do}\\
\textbf{\scriptsize7.} \> \> \> \> $c_1[u,v] \leftarrow c_2[v,u] \leftarrow c[u,v]$;\\
\textbf{\scriptsize8.} \> \> {\bf end-for}; \\
\textbf{\scriptsize9.} \> \> $c_1[i,\alpha] \leftarrow c_1[j,\alpha] \leftarrow c_2[i,\beta] \leftarrow c_2[j,\beta] \leftarrow |M|$;\\
\textbf{\scriptsize10.} \> \> \textit{Step \rom{2}. Looking for a violated CC (\ref{cc00})} \\
\textbf{\scriptsize11.} \> \> $\langle F_{s\to \alpha},\theta_{s\to \alpha} \rangle \leftarrow {\textit{max-flow}}_{s\to \alpha}{(\tilde{G}_1)}$; \\
\textbf{\scriptsize12.} \> \> \textbf{if }{$F_{s\to \alpha} < {\bar y}_i + {\bar y}_j$} \textbf{ then} \\
\textbf{\scriptsize13.} \> \> \> \>  $\mathcal{X} \leftarrow \mathcal{X} \cup \{ \langle N \backslash \theta_{s\to \alpha},\langle i,j \rangle \rangle\}$; \\
\textbf{\scriptsize14.} \> \> {\bf end-if}; \\
\textbf{\scriptsize15.} \> \> \textit{Step \rom{3}. Looking for a violated CC (\ref{cc01})} \\
\textbf{\scriptsize16.} \> \> $\langle F_{t\to \beta},\theta_{t\to \beta} \rangle \leftarrow {\textit{max-flow}}_{t\to \beta}{(\tilde{G}_2)}$; \\
\textbf{\scriptsize17.} \> \> \textbf{if }{$F_{t\to \beta} < {\bar y}_i + {\bar y}_j$} \textbf{ then} \\
\textbf{\scriptsize18.} \> \> \> \>  $\mathcal{X} \leftarrow \mathcal{X} \cup \{ \langle N\backslash \theta_{t\to \beta},\langle i,j \rangle \rangle\}$; \\
\textbf{\scriptsize19.} \> \> {\bf end-if}; \\
\textbf{\scriptsize20.} {\bf end-for}; \\
\textbf{\scriptsize21.} \textbf{return} $\mathcal{X}$;
\end{tabbing}
 }
\end{minipage}
}
}
\end{center}
\caption{Algorithm used to separate violated CCs.}
\label{cc_separation}
\end{figure}

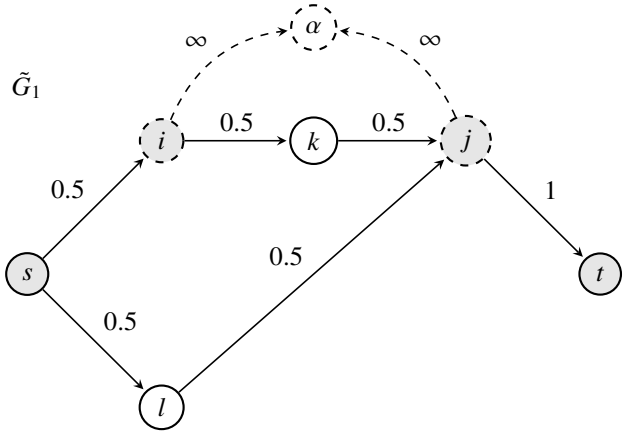
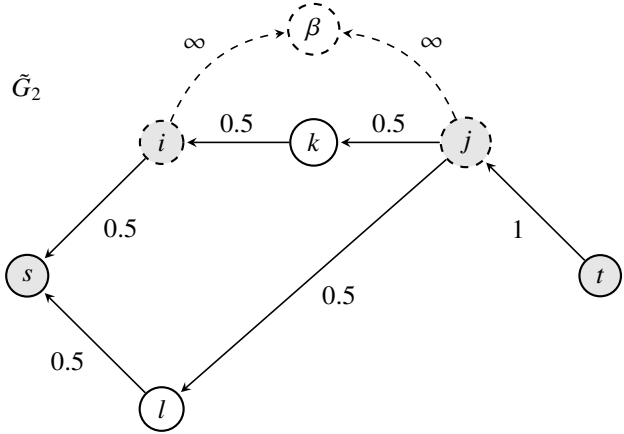
\begin{figure}
\centering

\begin{subfigure}{.5\textwidth}
  \centering
    \begin{tikzpicture}[
            > = stealth, % arrow head style
            shorten > = 1pt, % don't touch arrow head to node
            auto,
            node distance = 2.5cm, % distance between nodes
            semithick % line style
        ]

        \tikzstyle{every state}=[
            draw = black,
            thick,
            fill = white,
            minimum size = 4mm
        ]

        \node[state] (s) [fill = gray!20]{$s$};
        \node[state] (i) [above right of=s,fill = gray!20,dashed] {$i$};
        \node[state] (l) [below right of=s] {$l$};
        \node[state] (k) [right of=i,node distance = 2cm] {$k$};
        \node[state] (j) [right of=k,node distance = 2cm,fill = gray!20,dashed] {$j$};
        \node[state] (t) [below right of=j,fill = gray!20] {$t$};
        \node[state] (alpha) [dashed,above of=k,,node distance = 1.5cm]{$\alpha$};
        
        \node[state] (G) [above of = s, draw = none, node distance = 2.5cm]{$\tilde{G}_1$};
    
        \path[->] (s) edge node {0.5} (i);
        \path[->] (s) edge node {0.5} (l);
        \path[->] (l) edge node {0.5} (j);
        \path[->] (i) edge node {0.5} (k);
        \path[->] (k) edge node {0.5} (j);
        \path[->] (j) edge node {1} (t);
        
        \path[->] (i) edge[dashed,bend left] node {$\infty$} (alpha);
        \path[->] (j) edge[dashed,bend right] node[above right] {$\infty$} (alpha);
 %       \path[->] (l) edge node {0\,} (t);

    \end{tikzpicture}
    \caption{Example of an auxiliary graph $\tilde{G}_1$ used in the separation of CCs (\ref{cc00}).}
  \label{ex_ccs_a}
\end{subfigure}
\\
\begin{subfigure}{.5\textwidth}
  \centering
    \begin{tikzpicture}[
            > = stealth, % arrow head style
            shorten > = 1pt, % don't touch arrow head to node
            auto,
            node distance = 2.5cm, % distance between nodes
            semithick % line style
        ]

        \tikzstyle{every state}=[
            draw = black,
            thick,
            fill = white,
            minimum size = 4mm
        ]

        \node[state] (s) [fill = gray!20]{$s$};
        \node[state] (i) [above right of=s,fill = gray!20,dashed] {$i$};
        \node[state] (l) [below right of=s] {$l$};
        \node[state] (k) [right of=i,node distance = 2cm] {$k$};
        \node[state] (j) [right of=k,node distance = 2cm,fill = gray!20,dashed] {$j$};
        \node[state] (t) [below right of=j,fill = gray!20] {$t$};
        \node[state] (beta) [dashed,above of=k,,node distance = 1.5cm]{$\beta$};
        
        \node[state] (G) [above of = s, draw = none, node distance = 2.5cm]{$\tilde{G}_2$};
    
        \path[->] (i) edge node {0.5} (s);
        \path[->] (l) edge node {0.5} (s);
        \path[->] (j) edge node {0.5} (l);
        \path[->] (k) edge node[above] {0.5} (i);
        \path[->] (j) edge node[above] {0.5} (k);
        \path[->] (t) edge node {1} (j);
        
        \path[->] (i) edge[dashed,bend left] node {$\infty$} (beta);
        \path[->] (j) edge[dashed,bend right] node[above right] {$\infty$} (beta);
 %       \path[->] (l) edge node {0\,} (t);

    \end{tikzpicture}
  \caption{Example of an auxiliary graph $\tilde{G}_2$ used in the separation of CCs (\ref{cc01}).}
  \label{fig:ex_ccs_b}
\end{subfigure}
\caption{Auxiliary graphs built when considering the STOP instance of Figure~\ref{fig:ex}, the pair of conflicting vertices $\langle i,j \rangle$
and a fractional solution $\bar{x}$, with ${\bar x}_{si} = {\bar x}_{ik} = {\bar x}_{kj} = {\bar x}_{sl} = {\bar x}_{lj} = 0.5$, ${\bar x}_{lt} = 0$ and ${\bar x}_{jt} = 1$. Here, the values associated with the arcs are their corresponding capacities, and the infinity symbol stands for a sufficiently large value.}
\label{fig:ex_ccs}
\end{figure}

Once the auxiliary graphs are built for a given $\langle i,j \rangle \in \tilde{\mathcal{K}}$, the algorithm looks for violated CCs by solving two maximum flow problems: one from $s$ to $\alpha$ on $\tilde{G}_1$ and other from $t$ to $\beta$ on $\tilde{G}_2$. Let $\langle F_{s\to \alpha},\theta_{s\to \alpha} \rangle$ be the solution of the first maximum flow problem, i.e., $\langle F_{s\to \alpha},\theta_{s\to \alpha} \rangle = {\textit{max-flow}}_{s\to \alpha}{(\tilde{G}_1)}$.
Recall that $F_{s\to \alpha}$ gives the value of the resulting maximum flow, and $\theta_{s\to \alpha}$ defines a corresponding minimum cut, with $s \in \theta_{s \to \alpha}$. Then, the algorithm checks if $F_{s\to \alpha}$ is smaller than ${\bar y}_i + {\bar y}_j$. If that is the case, a violated CC
(\ref{cc00}) is identified and added to $\mathcal{X}$ (see lines 10-14, Figure~\ref{cc_separation}). Precisely, this inequality is denoted by
$\langle N\backslash \theta_{s\to \alpha},\langle i,j \rangle \rangle$ and defined as
\begin{equation}
\sum\limits_{e \in \delta^-(N \backslash \theta_{s \to \alpha})}{x_{e}} \geq y_{i} + y_j,
\end{equation}

\noindent where $N \backslash \theta_{s \to \alpha}$ corresponds to the subset $V \subseteq N\backslash \{s\}$ of (\ref{cc00}).

Likewise, let $\langle F_{t\to \beta},\theta_{t\to \beta} \rangle$ be the solution of the second maximum flow problem, i.e., $\langle F_{t\to \beta},\theta_{t\to \beta} \rangle = {\textit{max-flow}}_{t\to \beta}{(\tilde{G}_2)}$. If $F_{t\to \beta}$ is smaller than ${\bar y}_i + {\bar y}_j$, then a violated CC (\ref{cc01}) is identified and added to $\mathcal{X}$ (see lines 15-19, Figure~\ref{cc_separation}). This CC is denoted by $\langle N\backslash \theta_{t\to \beta},\langle i,j \rangle \rangle$ and defined as
\begin{equation}
\sum\limits_{e \in \delta^+(N \backslash \theta_{t \to \beta})}{x_{e}} \geq y_{i} + y_j.
\end{equation}

\noindent Here, $N \backslash \theta_{t \to \beta}$ corresponds to the subset $V \subseteq N\backslash \{t\}$ of (\ref{cc01}).

\subsection{LCIs}
\label{s_sep_lcis}
Here, we describe the separation procedure we adopt to obtain LCIs, which is based on the algorithmic framework of Gu et al.~\cite{Gu98}. Since the peculiarities behind lifting might be tricky, we first elucidate in details the concepts of \emph{down-lifting} and \emph{up-lifting}, which we apply throughout this section.

%and show how they are used to devise LCIs (\ref{lci3}) from classical cover inequalities (\ref{lci2}).
Consider a (not necessarily minimal) cover $C \subseteq P$ for the knapsack constraint (\ref{lci1}) and the corresponding cover inequality (\ref{lci2}). Moreover, let $\omega \in \mathbb{Z}^{|P|}$ be the coefficient vector of the $y$ variables in (\ref{lci2}), such that $\omega_i = 1$ for all $i \in C$, and $\omega_i = 0$ for all $i \in P\backslash C$. Accordingly, (\ref{lci2}) can be alternatively stated as
\begin{equation}
\sum\limits_{i \in P}{\omega_i y_i} \leq |C| -1. \label{lci2_2}
\end{equation}

The attempt to increase the coefficient values of the variables $y_i$ in (\ref{lci2_2}) such that $i \in C$ (and, thus, $\omega_i = 1$) is called down-lifting.
Likewise, the process of computing new coefficients for the variables $y_i$
such that $i \notin C$ (and, thus, $\omega_i = 0$) is called up-lifting.
In both cases, the aim is to strengthen the original cover inequality by replacing the initial coefficients with possibly greater positive integer values. Naturally, these lifted coefficients must be computed in a way that preserves the validity of the cover inequality with respect to the original knapsack constraint (\ref{lci1}).

In this work, both liftings are done sequentially (i.e., one variable at a time), and the lifted coefficients are computed by solving auxiliary knapsack problems to optimality.
In order to describe how down-lifting works, consider a partition of the cover $C$ into two disjoint sets $C_1$ and $C_2$, with $C_1 \neq \emptyset$. In this partition, $C_1$ keeps the indexes of the variables whose coefficients will not be updated (i.e., will remain equal to one), while $C_2$ identifies the variables to be down-lifted. Moreover, let $C'$ keep the indexes of the variables whose lifted coefficients were already computed, such that, initially, $C' = \emptyset$ and, by the end of all down-liftings, $C' = C_2$. During the down-lifting process, consider the following auxiliary inequality
\begin{equation}
\sum \limits_{i \in C_1}{y_{i}} + \sum \limits_{i \in C'}{\pi_iy_{i}} \leq |C_1| + \sum \limits_{i \in C'}{\pi_i} - 1, \label{lci4}
\end{equation}
\noindent where $\pi_i \in \mathbb{Z}$, $\pi_i \geq 1$ for all $i\in C'$, are the currently available lifted coefficients.
Without loss of generality, to down-lift a variable $y_j$, $j \in C_2\backslash C'$, the following auxiliary knapsack problem is solved
\begin{equation*}
\mbox{($\mathcal{A}_1$)\quad}\phi_1 = \max \Bigg\{ \sum \limits_{i \in C_1}{y_{i}} + \sum \limits_{i \in C'}{\pi_iy_{i}} :\, (\ref{lci1}),\, y_j = 0,\, y_i = 1 \;\forall\, i \in C_2 \backslash (C' \cup \{j\}),\, y \in \{0,1\}^{|P|} \Bigg\},
\end{equation*}

\noindent which consists of determining the maximum value $\phi_1$ that the left-hand side of (\ref{lci4}) can assume while satisfying the original knapsack constraint (\ref{lci1}) and fixing $y_{j} = 0$ and $y_i = 1$ for all $i \in C_2\backslash (C' \cup \{j\})$.
The resulting lifted coefficient of $y_j$ is given by the gap between $\phi_1$ and the current right-hand side value of (\ref{lci4}), i.e., $\pi_j = \phi_1 - (|C_1| + \sum \limits_{i \in C'}{\pi_i} - 1)$. Naturally, we update $C'\leftarrow C' \cup \{j\}$ after $y_j$ is lifted. 

Intuitively speaking, the process described above can be seen as removing $j$ from the cover (by setting $y_j = 0$) and, then, computing the maximum value $\pi_j$ can assume as to bring $j$ back to the cover. Recall that simply setting $\pi = \textbf{1}$ and considering $C' = C_2$ ensures the validity of (\ref{lci4}). Accordingly, keeping the remaining un-lifted variables $y_i$, for all $i \in C_2\backslash (C' \cup\{j\}$), fixed at one while down-lifting a variable $y_j$ yields the validity of the LCI obtained after all down-liftings are done (precisely, when $C' = C_2$). In particular, this resulting LCI takes the form
\begin{equation}
\sum \limits_{i \in C_1}{y_{i}} + \sum \limits_{i \in C_2}{\pi_iy_{i}} \leq |C_1| + \sum \limits_{i \in C_2}{\pi_i} - 1, \label{lci5}
\end{equation}

\noindent and is valid for the region $\Phi$ determined by the original knapsack constraint (\ref{lci1}), as defined by (\ref{lci1_2}). %In order to avoid unnecessary liftings, a word of caution must be given about down-lifting.
%\begin{proposition}
%\label{prop_down-lifting}
%If the cover $C$ is minimal, then performing down-lifting on any of the variables $y_j$, $j \in C$, of the original cover %inequality (\ref{lci2_2}) is ineffective, as it always leads to lifted coefficients equal to one.
%\end{proposition}
%
%Since we are not sure if this result is already proven in the literature, we devise a formal proof for it in \ref{appendix_a}.

Now, let the set $D \subseteq P \backslash C$ identify the indexes of the $y$ variables to be up-lifted.
Moreover, let $C''$ be the index set of the variables in $P \backslash C$ whose coefficients are already established, such that, initially, $C'' = P \backslash (C \cup D)$ and, by the end of the up-liftings, $C'' = P \backslash C$.
%In this sense, $\mu_i \in \mathbb{Z}$, $\mu_i \geq 0$ for all $i\in P \backslash (C \cup D)$, are the initially established coefficients.
Notice that, if $D = P \backslash C$, then $C''$ is empty before the up-liftings take place.
During the whole up-lifting process, consider the following inequality
\begin{equation}
\sum \limits_{i \in C_1}{y_{i}} + \sum\limits_{i \in C_2}{\pi_iy_i} + \sum\limits_{i \in C''}{\mu_iy_i} \leq |C_1| + \sum\limits_{i \in C_2}{\pi_i} -1, \label{lci6}
\end{equation}

\noindent where $\mu_i \in \mathbb{Z}$, $\mu_i \geq 0$ for all $i\in C''$, are the currently available coefficients. One may note that simply setting $\mu = \textbf{0}$, i.e., removing all the variables $y_i$, $i \in P \backslash C$, from (\ref{lci6}), turns this inequality into (\ref{lci5}). Moreover, considering $\mu = \textbf{0}$ and $\pi = \textbf{1}$ reduces (\ref{lci6}) to the original cover inequality (\ref{lci2_2}). %In fact, differently from the auxiliary inequality (\ref{lci4}) related to down-lifting, (\ref{lci6}) is valid for the region in (\ref{lci1_2}) at any point of the up-lifting process, 

Without loss of generality, to up-lift a variable $y_k$, $k \in D\backslash C''$, the following knapsack problem is solved
\begin{equation*}
\mbox{($\mathcal{A}_2$)\quad}\phi_2 = \max \Bigg\{ \sum \limits_{i \in C_1}{y_{i}} + \sum\limits_{i \in C_2}{\pi_iy_i} + \sum\limits_{i \in C''}{\mu_iy_i} :\, (\ref{lci1}),\, y_k = 1,\, y \in \{0,1\}^{|P|} \Bigg\}.
\end{equation*}
%\noindent which consists of determining the maximum value $\phi_2$ that the left-hand side of the inequality (\ref{lci6}) can assume while satisfying the original knapsack constraint (\ref{lci1}) and fixing $y_{k} = 1$.
\noindent In practice, problem $\mathcal{A}_2$ can be seen as forcefully adding $k$ to the cover (by setting $y_k = 1$) and, then, computing the maximum value $\phi_2$ 
that the left-hand side of the current inequality (\ref{lci6}) can assume while satisfying the original knapsack constraint (\ref{lci1}). Notice that, as $k$ does not belong to the cover originally, satisfying (\ref{lci1}) and $y_k = 1$ can only lead to a value $\phi_2$ inferior or equal to the current right-hand side of (\ref{lci6}). Thus,
the resulting lifted coefficient of $y_k$ is given by the gap between the right-hand side value of (\ref{lci6}) and $\phi_2$. Under the current circumstances, we would have $\mu_k = (|C_1| + \sum \limits_{i \in C_2}{\pi_i} - 1) - \phi_2$. Also notice that, if $p_k > \lfloor\tau\rfloor$, problem $\mathcal{A}_2$ becomes infeasible. In this case, the original coefficient is preserved, i.e., 
$\mu_k$ is set to $\omega_k = 0$. Here, we update $C''\leftarrow C'' \cup \{k\}$ once $y_k$ is lifted.

In this work, up-lifting is also applied to lift inequalities that are only valid for a restricted region of $\Phi$  --- see (\ref{lci1_2}). In particular, consider an inequality of the form 
\begin{equation}
\sum\limits_{i \in C_1}{y_i} + \sum\limits_{i \in C''}{\mu_iy_i}\leq |C_1| - 1 \label{lci_up2}
\end{equation}
\noindent that is valid for the restricted polyhedron
\begin{equation}
\Phi' = \Big\{y \in \{0,1\}^{|P\backslash C_2|}:\, \sum \limits_{i \in P \backslash C_2}{p_{i}y_{i}} \leq \lfloor\tau\rfloor - \sum\limits_{i \in C_2}{p_i} \Big\}.
\end{equation}
\noindent Notice that $\Phi'$ corresponds to $\{y \in \Phi:\, y_i = 1 \; \forall\, i \in C_2\}$.
Then, in this case, to up-lift a variable $y_k$, $k \in D\backslash C''$, the following knapsack problem is solved
\begin{equation*}
\phi_2' = \max \Bigg\{ \sum \limits_{i \in C_1}{y_{i}} + \sum\limits_{i \in C''}{\mu_iy_i} :\, (\ref{lci1}),\, y_k = 1,\, y_i = 1\; \forall i\, \in C_2,\,y \in \{0,1\}^{|P|} \Bigg\},
\end{equation*}
\noindent and the resulting lifted coefficient assumes the value $\mu_k = (|C_1| - 1) - \phi_2'$. Here, setting $y_i = 1$ for all $i \in C_2$ yields the validity of the resulting LCI with respect to the unrestricted region $\Phi$. In fact, in this specific case of up-lifting, the inequality
\begin{equation}
\sum \limits_{i \in C_1}{y_{i}} + \sum\limits_{i \in C_2}{y_i} + \sum\limits_{i \in C''}{\mu_iy_i} \leq |C_1| + |C_2| -1
\end{equation}
\noindent is always valid for $\Phi$ and corresponds to (\ref{lci6}) when $\pi = \textbf{1}$.

For simplicity, we described the process of down-lifting by considering a classical cover inequality of type (\ref{lci2}). However, this kind of lifting is usually applied only after some or all the variables identified in $P \backslash C$ have already been up-lifted. In this sense, consider the following inequality obtained from up-lifting the variables $y_i$, for all $i \in D \subseteq P \backslash C$, of either (\ref{lci2_2}) or (\ref{lci_up2})
\begin{equation}
\sum \limits_{i \in C_1}{y_{i}} + \sum \limits_{i \in C_2}{y_{i}} + \sum\limits_{i \in D}{\mu_iy_i} \leq |C_1| + |C_2| -1. \label{lci7}
\end{equation}
\noindent Recall that $C_1$ and $C_2$ compose a partition of the cover $C$, with $C_1 \neq \emptyset$, and $C' \subseteq C_2$ keeps the indexes of the variables whose down-lifted coefficients were already computed. 
In this case, the auxiliary inequality used during the down-lifting of a variable $y_j$, $j \in C_2\backslash C'$, would take the form
\begin{equation}
\sum \limits_{i \in C_1}{y_{i}} + \sum \limits_{i \in C'}{\pi_iy_{i}} + \sum\limits_{i \in D}{\mu_iy_i} \leq |C_1| + \sum \limits_{i \in C'}{\pi_i} - 1, \label{lci8}
\end{equation}
\noindent while the auxiliary knapsack problem solved would be
\begin{equation*}
\phi'_1 = \max \Bigg\{ \sum \limits_{i \in C_1}{y_{i}} + \sum \limits_{i \in C'}{\pi_iy_{i}} + \sum\limits_{i \in D}{\mu_iy_i}:\, (\ref{lci1}),\, y_j = 0,\, y_i = 1 \;\forall\, i \in C_2 \backslash (C' \cup \{j\}),\, y \in \{0,1\}^{|P|} \Bigg\}.
\end{equation*}
\noindent Furthermore, the lifted coefficient would be set to $\pi_j = \phi'_1 - (|C_1| + \sum \limits_{i \in C'}{\pi_i} - 1)$.

Taking into account the liftings detailed above, we describe in Figure~\ref{lci_separation} the algorithm used to separate LCIs from a fractional solution $(\bar{x},\bar{y},\bar{f},\bar{\varphi})$. At the first step of the algorithm, we look for a cover $C$ for the corresponding knapsack constraint (\ref{lci1}). Starting from $C = \emptyset$, we sequentially add to $C$ elements of the subset
$\{i \in P:\,{\bar y}_i > 0\}$ in non-increasing order of the corresponding values in $\bar{y}$.
This is done until $C$ covers (\ref{lci1}), i.e., $\sum\limits_{i \in C}{p_i} > \lfloor\tau\rfloor$, or until there are no more elements left to be added (see lines 1-8, Figure~\ref{lci_separation}). 
If no cover is found, the algorithm terminates (line 9, Figure~\ref{lci_separation}). Otherwise, the cover found is converted into a minimal one by deleting elements from it. This conversion prioritizes the deletion of elements with the smallest relaxation values (see lines 10-15, Figure~\ref{lci_separation}), as a way to increase the chances of devising a violated LCI. 
\begin{figure}[!ht]
\begin{center}
\scalebox{1}
{
\framebox
{
\begin{minipage}[t]{25cm}
{\small
\begin{tabbing}
xxx\=xxx\=xxx\=xxx\=xxx\=xxx\=xxx\=xxx\=xxx\=xxx\= \kill
\textbf{Input: }A fractional solution $(\bar{x},\bar{y},\bar{f},\bar{\varphi})$ and its corresponding bound $\tau = \sum\limits_{i \in P}{p_i{\bar y}_i}$. \\
\textbf{Output: }{An LCI of type (\ref{lci3}), if any.} \\
\textbf{\scriptsize1.} \textit{Step \rom{1}. Finding an initial cover}\\
\textbf{\scriptsize2.} Set the initial cover $C \leftarrow \emptyset$ and the iterator $j \leftarrow 1$;\\
\textbf{\scriptsize3.} Create a vector $F$ with the indexes in $\{i \in P:\,{\bar y}_i > 0\}$;\\
\textbf{\scriptsize4.} Sort $F$ in non-increasing order of $\bar{y}$;\\
\textbf{\scriptsize5.} \textbf{while} {$j \leq |F|$} \textbf{and} {$\sum\limits_{i \in C}{p_i} \leq \lfloor\tau\rfloor$} \textbf{do}\\
\textbf{\scriptsize6.} \> \> $C \leftarrow C \cup \{F[j]\}$;\\
\textbf{\scriptsize7.} \> \> $j \leftarrow j + 1$;\\
\textbf{\scriptsize8.} {\bf end-while}; \\
\textbf{\scriptsize9.} \textbf{if} {$\sum\limits_{i \in C}{p_i} \leq \lfloor\tau\rfloor$} \textbf{then} halt with no resulting LCI;\\
\textbf{\scriptsize10.} \textit{Step \rom{2}. Converting the cover into a minimal one}\\
\textbf{\scriptsize11.} Create a vector $\tilde{C}$ with the elements in $C$;\\
\textbf{\scriptsize12.} Sort $\tilde{C}$ in non-decreasing order of $\bar{y}$;\\
\textbf{\scriptsize13.} \textbf{for} {$k \leftarrow 1,\dots,|\tilde{C}|$} \textbf{do}\\
\textbf{\scriptsize14.} \> \> \textbf{if} {$\sum\limits_{i \in C\backslash\{\tilde{C}[k]\}}{p_i} > \lfloor\tau\rfloor$} \textbf{then} {$C \leftarrow C \backslash \{\tilde{C}[k]\}$};\\
\textbf{\scriptsize15.} {\bf end-for}; \\
\textbf{\scriptsize16.} \textit{Step \rom{3}. Lifting}\\
\textbf{\scriptsize17.} Define $Q = \{i \in C:\, \bar{y}_i = 1\}$, $\bar{C}_1 = \{i \in P\backslash C:\, \bar{y}_i > 0\}$ and $\bar{C}_2 = \{i \in P\backslash C:\, \bar{y}_i = 0\}$;\\
\textbf{\scriptsize18.} Define the initial inequality $\sum\limits_{i \in C\backslash Q}{y_i} \leq |C\backslash Q| -1$;\\
\textbf{\scriptsize19.} Up-lift $y_i$, for all $i \in {\bar C}_1$, obtaining $\sum\limits_{i \in C\backslash Q}{y_i} +\sum\limits_{i \in {\bar C}_1}{\mu_i y_i} \leq |C\backslash Q| -1$;\\
\textbf{\scriptsize20.} Down-lift $y_i$, for all $i \in Q$, obtaining $\sum\limits_{i \in C\backslash Q}{y_i} + \sum\limits_{i \in Q}{\pi_i y_i} + \sum\limits_{i \in {\bar C}_1}{\mu_i y_i} \leq |C\backslash Q| + \sum\limits_{i \in Q}{\pi_i} -1$;\\
\textbf{\scriptsize21.} Up-lift $y_i$, for all $i \in \bar{C}_2$, obtaining $\sum\limits_{i \in C\backslash Q}{y_i} + \sum\limits_{i \in Q}{\pi_i y_i} + \sum\limits_{i \in {P \backslash C}}{\mu_i y_i} \leq |C\backslash Q| + \sum\limits_{i \in Q}{\pi_i} -1$;\\
\textbf{\scriptsize22.} \textbf{return} the resulting LCI;
\end{tabbing}
 }
\end{minipage}
}
}
\end{center}
\caption{Algorithm used to separate possibly violated LCIs.}
\label{lci_separation}
\end{figure}

In order to guide the lifting process, we define three subsets of indexes, denoted by
$Q = \{i \in C:\, \bar{y}_i = 1\}$, $\bar{C}_1 = \{i \in P\backslash C:\, \bar{y}_i > 0\}$ and $\bar{C}_2 = \{i \in P\backslash C:\, \bar{y}_i = 0\}$. Note that $\bar{C}_1 \cup \bar{C}_2 = P \backslash C$ and $\bar{C}_1 \cap \bar{C}_2 = \emptyset$. The following result guarantees that $C \backslash Q \neq \emptyset$ at this point of the separation procedure, i.e., there is at least one element in $C$ whose corresponding value in $\bar{y}$ is fractional.

\begin{proposition}
\label{prop_minimal_cover}
Given a fractional solution $(\bar{x},\bar{y},\bar{f},\bar{\varphi})$ for $\mathcal{L}_2$ and a bound $\tau = \sum\limits_{i \in P}{p_i{\bar y}_i}$, then any cover $C$ for the corresponding knapsack constraint (\ref{lci1}), with ${\bar y}_i > 0$ for all $i \in C$, necessarily has an element $i \in C$ such that $0 < {\bar y}_i < 1$.
\end{proposition}

\begin{proof}
Consider a cover $C$ for the knapsack constraint (\ref{lci1}), with ${\bar y}_i > 0$ for all $i \in C$, and suppose, by contradiction, that ${\bar y}_i = 1$ for all $i \in C$.
In this case, since $p_i \in \mathbb{Z}^+$ for all $i \in P$ (by definition), and $C$ covers (\ref{lci1}), we would have
\begin{equation}
\lfloor\tau\rfloor = \lfloor\sum\limits_{i \in P}{p_i{\bar y}_i}\rfloor \geq \lfloor \sum\limits_{i \in C}{p_i{\bar y}_i}\rfloor = \sum\limits_{i \in C}{p_i} > \lfloor\tau\rfloor,
\end{equation}

\noindent which is a contradiction. \qedhere
\end{proof}

Then, starting from the inequality
\begin{equation}
\sum\limits_{i \in C\backslash Q}{y_i} \leq |C\backslash Q| -1, \label{lci11}
\end{equation}
\noindent we up-lift the variables $y_i$, for all $i \in \bar{C}_1$, which yields the inequality 
\begin{equation}
\sum\limits_{i \in C\backslash Q}{y_i} +\sum\limits_{i \in {\bar C}_1}{\mu_i y_i} \leq  |C\backslash Q| -1. \label{lci12}
\end{equation}
\noindent Since (\ref{lci11}) only considers a restricted region of the original polyhedron $\Phi$, neither (\ref{lci11}) nor
(\ref{lci12}) are valid LCIs. Nevertheless, the inequality
\begin{equation}
\sum\limits_{i \in C}{y_i} + \sum\limits_{i \in {\bar C}_1}{\mu_i y_i} \leq |C| -1
\end{equation}

\noindent is valid at this point and can be strengthened by down-lifting the variables $y_i$, for all $i \in Q$. Once these down-liftings are performed, the LCI takes the form
\begin{equation}
\sum\limits_{i \in C\backslash Q}{y_i} + \sum\limits_{i \in Q}{\pi_i y_i} + \sum\limits_{i \in {\bar C}_1}{\mu_i y_i} \leq |C\backslash Q| + \sum\limits_{i \in Q}{\pi_i} -1.
\end{equation}

At last, we up-lift the remaining variables $y_i$, for all $i \in \bar{C}_2$, obtaining
\begin{equation}
\sum\limits_{i \in C\backslash Q}{y_i} + \sum\limits_{i \in Q}{\pi_i y_i} + \sum\limits_{i \in {P \backslash C}}{\mu_i y_i} \leq |C\backslash Q| + \sum\limits_{i \in Q}{\pi_i} -1,
\end{equation}

\noindent with $\mu \geq \textbf{0}$ and $\pi \geq \textbf{1}$.
The sequence of liftings detailed above is summarized at the last step of the algorithm in Figure~\ref{lci_separation} (see lines 16-21).
Naturally, the LCI returned by the separation algorithm (line 22, Figure~\ref{lci_separation}) is only used in the cutting-plane scheme if it is violated by $(\bar{x},\bar{y},\bar{f},\bar{\varphi})$, i.e., if $\sum\limits_{i \in C\backslash Q}{{\bar y}_i} + \sum\limits_{i \in Q}{\pi_i {\bar y}_i} + \sum\limits_{i \in {P \backslash C}}{\mu_i {\bar y}_i} > |C\backslash Q| + \sum\limits_{i \in Q}{\pi_i} -1$.

\section{Baseline branch-and-cut algorithm}
\label{s_branch-and-cut}
 The exact algorithm of Bianchessi et al.~\cite{Bianchessi2017} for TOP is used as a baseline to evaluate the performance of the cutting-plane scheme here proposed. In the case of STOP, the compact formulation $\mathcal{F}_1$, defined by (\ref{b100})-(\ref{b113}), with the addition of (\ref{ineq00}), is solved by means of an optimization solver (in this case, CPLEX) in a way that GCCs are separated on the fly at each node of the branch-and-bound tree. In order to avoid a tailing-off phenomenon, the separation of cuts at each node of the tree is iterated until the bound improvement is no greater than a pre-established tolerance $\epsilon_1$.
 In this branch-and-cut algorithm, the same procedure described in Section~\ref{s_separate_gccs} is applied to separate GCCs for each fractional solution, and all the violated cuts found are added to the model. 
 
 We remark that, although (\ref{ineq00}) is redundant for $\mathcal{F}_1$ (see, again, Corollary~\ref{corol01}),
 this inequality was preserved in our experiments as a way to properly reproduce the original algorithm of Bianchessi et al. \cite{Bianchessi2017}. In fact, we conjecture that CPLEX does benefit from this inequality when the separation of built-in cuts is enabled. This is quite intuitive, since (\ref{ineq00}) takes the form of a knapsack constraint, and cover cuts are among the several classical valid inequalities that compose the built-in cuts.

\section{Cutting-plane scheme}
\label{s_cutting-plane}
When solving (mixed) integer problems, cutting-plane algorithms work by iteratively reinforcing an initial Linear Programming (LP) model, which usually corresponds to the linearly relaxed version of the original problem. Precisely, at each iteration, the algorithm seeks linear inequalities that are violated by the solution of the current LP model.
These inequalities are referred to as \emph{cuts} and are added to the model on the fly until a stopping condition is met or the current solution is feasible (and, thus, optimal) for the original integer problem.

The algorithm here proposed solves the compact formulation $\mathcal{F}_2$ within a cutting-plane scheme that starts from the LP model $\mathcal{L}_2$, the linear relaxation of $\mathcal{F}_2$. Accordingly, the polyhedron defined by $\mathcal{L}_2$ is gradually restricted by the addition of new linear inequalities, which, in this case, correspond to the valid inequalities detailed in Section~\ref{s_cuts}.
This initial step, called \emph{cutting-plane phase}, is summarized in Figure~\ref{fig_cutting-plane} and detailed in the sequel.

Consider the polyhedron $\Omega$ defined by the feasible region of $\mathcal{L}_2$. Precisely, $\Omega = \big\{(x,y,f,\varphi) \in \mathbb{R}^{|A|}\times\mathbb{R}^{|N|}\times\mathbb{R}^{|A|}\times\mathbb{R}:\, (\ref{b101})$-$(\ref{b105}),\,(\ref{csc107})$-$(\ref{csc110}),\,(\ref{csc114}),\, {\bf 0}\leq x \leq {\bf 1},\, {\bf 0}\leq y \leq {\bf 1} \mbox{ and } f \geq \textbf{0}\big\}$. For simplicity, assume that $\mathcal{F}_2$ is feasible, and, thus, $\Omega \neq \emptyset$.
Let LP$^{\psi}$ and UB$^{\psi}$ keep, respectively, the LP model and the dual (upper) bound on the optimal solution of $\mathcal{F}_2$ available by the end of an iteration $\psi$ of the cutting-plane phase. Likewise, the set $\Gamma^{\psi}$ keeps all the linear inequalities added to the original LP $\mathcal{L}_2$ until the iteration $\psi$. In addition, \emph{opt}(LP$^{\psi}$) denotes the optimal solution value of a model LP$^{\psi}$.
Here, the iteration $\psi = 0$ stands for the initialization of the cutting-plane phase.
Accordingly, the initial model and its corresponding bound are denoted by LP$^{0}$ and UB$^{0}$, respectively, and $\Gamma^{0} = \emptyset$.
\begin{figure}[!ht]
\begin{center}
\scalebox{1}
{
\framebox
{
\begin{minipage}[t]{25cm}
{\small
\begin{tabbing}
xxx\=xxx\=xxx\=xxx\=xxx\=xxx\=xxx\=xxx\=xxx\=xxx\= \kill
\textbf{Input: }{The initial LP model $\mathcal{L}_2$ and a tolerance value $\epsilon_2$.} \\
\textbf{Output: }{A reinforced LP model and a dual bound on the optimal solution of $\mathcal{F}_2$.} \\
\textbf{\scriptsize1.} Initialize the iterator $\psi \leftarrow 0$;\\
\textbf{\scriptsize2.} Define LP$^{0} = \mathcal{L}_2$ and set $\Gamma^{0} \leftarrow \emptyset$;\\
\textbf{\scriptsize3.} UB$^{0} \leftarrow$ \emph{opt}(LP$^{0}$);\\
\textbf{\scriptsize4.} \textbf{do}\\
\textbf{\scriptsize5.} \> \> Update $\psi \leftarrow \psi + 1$;\\
\textbf{\scriptsize6.} \> \> Set $\Gamma^{\psi} \leftarrow \Gamma^{\psi-1}$;\\
\textbf{\scriptsize7.} \> \> Separate and add to $\Gamma^{\psi}$ \emph{some} of the violated GCCs, if any;\\
\textbf{\scriptsize8.} \> \> Separate and add to $\Gamma^{\psi}$ \emph{some} of the violated CCs, if any;\\
\textbf{\scriptsize9.} \> \> Separate and add to $\Gamma^{\psi}$ a violated LCI, if any;\\
\textbf{\scriptsize10.} \> \> Define LP$^{\psi} = \Big\{\max \sum\limits_{j \in P}{p_jy_j}:$ \mbox{inequality $i$ is satisfied for all} $i\in \Gamma^{\psi} ,\,(x,y,f,\varphi) \in \Omega\Big\}$;\\
\textbf{\scriptsize11.} \> \> UB$^{\psi} \leftarrow$ \emph{opt}(LP$^\psi$);\\
\textbf{\scriptsize12.} {\bf while} {($\Gamma^{\psi}\neq\Gamma^{\psi-1}$)} \textbf{and} (UB$^{\psi-1}$- UB$^{\psi} > \epsilon_2$); \\
\textbf{\scriptsize13.} \textbf{return} (LP$^{\psi}$, UB$^{\psi}$);
\end{tabbing}
 }
\end{minipage}
}
}
\end{center}
\caption{Description of the cutting-plane phase of the algorithm proposed.}
\label{fig_cutting-plane}
\end{figure}

After the initialization (lines 1-3, Figure~\ref{fig_cutting-plane}), the iterative procedure takes place. At each loop of the cutting-plane phase, the iterator $\psi$ is updated, and the set $\Gamma^{\psi}$ is initialized with the cuts found so far (see lines 5 and 6, Figure~\ref{fig_cutting-plane}). Then, the algorithm looks for linear inequalities violated by the solution of the current model, which, at this point, corresponds to LP$^{\psi-1}$. These cuts are found by means of the separation procedures described in Section~\ref{s_separation}. Instead of selecting all the violated cuts found, we only add to $\Gamma^{\psi}$ the most violated cut (if any) and the ones that are sufficiently orthogonal to it. As verified in several works (see, e.g., \cite{Wesselmann2012,Samer2015,Bicalho2016}), this strategy is able to balance the strength and diversity of the cuts separated, while limiting the model size. Here, this strategy is applied to select both GCCs and CCs, but separately (lines 7 and 8, Figure~\ref{fig_cutting-plane}). Naturally, this filtering procedure does not apply to LCIs, since at most a single LCI is separated per iteration (line 9, Figure~\ref{fig_cutting-plane}). Details on how these cuts are selected are given in Section~\ref{s_implementation_details}.

After looking for violated inequalities (cuts), we define an updated model LP$^{\psi}$, which corresponds to adding to $\mathcal{L}_2$ all the cuts selected so far (see line 10, Figure~\ref{fig_cutting-plane}). Accordingly, the current bound is set to the optimal solution value of LP$^{\psi}$ (line 11, Figure~\ref{fig_cutting-plane}). The algorithm iterates until either no more violated cuts are found or the bound improvement of the current model is inferior or equal to the tolerance $\epsilon_2$ (see line 12, Figure~\ref{fig_cutting-plane}). We highlight that the order in which the three types of inequalities are separated is not relevant in this case (lines 7-9, Figure~\ref{fig_cutting-plane}), as the updated LP model is not solved until all separation procedures are done.

Once the cutting-plane phase is over, the integrality of the variables $x$ and $y$ is restored, and the resulting reinforced model is solved to optimality by an optimization solver. At this point, inequalities (\ref{csc110}) from $\mathcal{F}_2$ are turned into cuts (instead of actual restrictions), just like the ones selected in the cutting-plane phase. As discussed by the end of Section~\ref{s_models}, the aim is to take advantage from cut management mechanisms within the optimization solver's branch-and-cut scheme.

In practical terms, the algorithm described above is a branch-and-cut in which GCCs, CCs and LCIs are only separated at the root node of the branch-and-bound tree. In fact, the cutting-plane proposed can be easily extended by performing the cutting-plane phase of Figure~\ref{fig_cutting-plane} at each node of the branch-and-bound tree. Pilot experiments suggested that such approach is not worthy in this case, as the additional strengthening of the model does not always pay off the loss in compactness. Nevertheless, investigating other strategies of cut selection and management might lead to more promising results.

\section{Implementation details}
\label{s_implementation_details}
%In this section, we discuss some of the implementation choices made in this work.
All the codes were developed in C++, along with the optimization solver ILOG CPLEX 12.6.
The baseline branch-and-cut algorithm described in Section~\ref{s_branch-and-cut} was implemented using the callback mechanism of CPLEX. Moreover, CPLEX was used to solve the LP models within the cutting-plane phase of the cutting-plane algorithm proposed and to close the integrality gap of the reinforced MILP model obtained from the addition of cuts. We kept the default configurations of CPLEX in our implementations, since all the previous works in the literature of TOP that used CPLEX make the same choice.

Regarding the separation of cuts, we solved maximum flow sub-problems with the implementation of the preflow push-relabel algorithm of Goldberg and Tarjan~\cite{Goldberg1988} provided by the open-source Library for Efficient Modeling and Optimization in Networks --- LEMON~\cite{Dezso2011}. The knapsack sub-problems that arise during the separation of LCIs were solved through classical dynamic programming based on Bellman recursion \cite{Bellman57}.

In the selection of cuts, we adopted the \emph{absolute violation} criterion to determine which inequalities are violated by a given solution. In turn, the so-called \emph{distance} criterion was used to properly compare two cuts, i.e., to determine which one is most violated by a solution. Given an $n$-dimensional column vector $w$ of binary variables, a point $\bar{w} \in \mathbb{R}^n$ and an inequality of the general form $a^T w \leq b$, with $a \in \mathbb{R}^n$, $b \in \mathbb{R}$, the absolute violation of this inequality with respect to $\bar{w}$ is simply given by $a^T \bar{w} - b$. Moreover, the distance from $a^T w \leq b$ to $\bar{w}$ corresponds to the Euclidean distance between the hyperplane $a^T w = b$ and $\bar{w}$, which is equal to $\frac{(a^T \bar{w} - b)}{\lVert a \rVert}$, where $\lVert a \rVert$ is the Euclidean norm of $a$.

In our implementation of the cutting-plane algorithm, we set two parameters for each type of inequality separated: a precision one used to classify the inequalities into violated or not (namely \emph{absolute violation precision}), and another one to discard cuts that are not sufficiently orthogonal to the most violated ones. The latter parameter determines the minimum angle that an arbitrary cut must form with the most violated cut, as not to be discarded. In practice, this parameter establishes the maximum acceptable inner product between the arbitrary cut and the most violated one. Accordingly, we call it the \emph{maximum inner product}. In the case of two inequalities $a_1^T w \leq b_1$ and $a_2^T w \leq b_2$, with $a_1,a_2 \in \mathbb{R}^n$ and $b_1,b_2 \in \mathbb{R}$, the inner product between them is given by $\frac{(a_1^Ta_2)}{\lVert a_1 \rVert \lVert a_2 \rVert}$ and corresponds to the cosine of the angle defined by them.
The values of these parameters were set according to pilot experiments and are shown in Table~\ref{table_parameters}.
\begin{table}[!ht]
\center
\caption{Parameter configuration adopted in the separation and selection of valid inequalities in the cutting-plane algorithm.}
\label{table_parameters}
\begin{tabular}{llllll}
\toprule
& &  \multicolumn{2}{c}{\textbf{Parameter}} \\
\cmidrule{3-4}
{\bf Inequalities} &  & Absolute violation precision & Maximum inner product\\
\midrule
GCCs & & 0.05 & 0.03\\
CCs & & 0.3 & 0.03 \\
LCIs & & $10^{-5}$  & ~~--\\
\bottomrule
\end{tabular}
\end{table}

The tolerance input value $\epsilon_2$ of the cutting-plane algorithm (see Figure~\ref{fig_cutting-plane}) was set to $10^{-3}$. In the case of the baseline branch-and-cut, the absolute violation precision regarding the separation of GCCs was also set to 0.05. In addition, the tailing-off tolerance $\epsilon_1$ was set to $10^{-3}$, the same value adopted to $\epsilon_2$.

We remark that all the parameter configurations described above were established according to pilot tests on a control set of 10 instances, composed of both challenging instances and some of the smallest ones. This control set is detailed in Table~\ref{table_control_instances}, where we report, for each instance, the number of vertices ($|N|$), the number of vehicles ($|M|$) and the route duration limit ($T$). The reduced number of instances was chosen as a way to avoid overfitting.

\begin{table}[!ht]
\center
\caption{Control set of TOP instances used to tune the algorithm's parameters.}
\label{table_control_instances}
\begin{tabular}{lrrr}
\toprule
\textbf{Instance} & \multicolumn{1}{c}{$|N|$} & \multicolumn{1}{c}{$|M|$} & \multicolumn{1}{c}{$T$} \\
\midrule
p3.3.r & 33 & 3 & 33.3 \\
p4.3.j & 100 & 3 & 46.7 \\
p4.3.n & 100 & 3 & 60.0 \\
p5.3.m & 66 & 3 & 21.7 \\
p5.3.r & 66 & 3 & 30.0 \\
p6.2.k & 64 & 2 & 32.5 \\
p6.3.m & 64 & 3 & 25.0 \\
p6.3.n & 64 & 3 & 26.7 \\
p7.3.o & 102 & 3 & 100.0 \\
p7.3.p & 102 & 3 & 106.7 \\
\bottomrule
\end{tabular}
\end{table}

We also highlight that, despite of their exact-like form, the algorithms adopted to separate GCCs and CCs are heuristics. In particular, notice that the simple fact that they only consider the cuts that are violated by at least a constant factor makes them heuristics in practice. Moreover, these algorithms adopt a stopping condition based on bound improvement of subsequent iterations, which might halt the separation before all the violated cuts are found.
Then, in practical terms, the separation algorithms adopted do not necessarily give the actual theoretical bounds obtained from the addition of the inequalities proposed.

\section{Computational experiments}
\label{s_experiments}
The computational experiments were performed on a 64 bits Intel Core i7-4790K machine with 4.0 GHz and 15.0 GB of RAM, under Linux operating system. The machine has four physical cores, each one running at most two threads in hyper-threading mode. Here, the Baseline Branch-and-Cut and the Cutting-Plane Algorithm are referred to as B-B\&C and CPA, respectively. Both of them were set to run for up to 7200s, the same time limit established in previous works concerning TOP \cite{Poggi10,Boussier07, Dang13b, Keshtkaran16, ElHajj2016, Bianchessi2017}.

In our experiments, we used the benchmark of TOP instances introduced by Chao et al.~\cite{Chao96}, which consists of complete graphs with up to 102 vertices. In this case, no mandatory vertices are considered.
Based on this benchmark, we also generated new instances by randomly setting a percentage of the vertices as mandatory.
Here, this percentage was set to only 5\%, as greater values led to the generation of too many infeasible instances.

The original benchmark of Chao et al.~\cite{Chao96} is composed of 387 instances, which are divided into seven data sets, according to the number of vertices of their graphs. In each set, the instances only differ by the time limit imposed on the route duration and the number of vehicles, which varies from 2 to 4. The characteristics of these data sets are detailed in table~\ref{table_instances}. For each set, it is reported the number of instances (\#), the number of vertices in the graphs ($|N|$) and the range of values that the route duration limit $T$ assumes. 
\begin{table}[!ht]
\center
\caption{Description of the original benchmark of TOP instances.}
\label{table_instances}
\begin{tabular}{cccccccc}
\toprule
{\bf Set} & {1} & 2 & 3 & 4 & 5 & 6 & 7\\
\midrule
\# & 54 & 33 & 60 & 60 & 78 & 42 & 60\\
$|N|$ & 32 & 21 & 33 & 100 & 66 & 64 & 102\\
$T$  & 3.8--22.5 & 1.2--42.5 & 3.8--55 & 3.8--40 & 1.2--65 & 5--200 & 12.5--120 \\
\bottomrule
\end{tabular}
\end{table}

As done in previous works \cite{Dang13b,Bianchessi2017}, we pre-processed all the instances used in our experiments by removing vertices and arcs that are inaccessible with respect to the limit $T$ imposed on the total traverse times of the routes. To this end, we considered the $R$ matrix defined by the end of Section~\ref{s_notation}, which keeps, for each pair of vertices, the time duration of a minimum time path between them.
Moreover, in the specific case of the cutting-plane algorithm, constraints (\ref{b104}) are implicitly satisfied by deleting all the arcs that either enter the origin $s$ or leave the destination $t$. Naturally, the time spent in these pre-processings are included in the execution times of the algorithms tested.

In Section~\ref{s_results_top}, we compare the performance of CPA with B-B\&C and other exact algorithms in the literature of TOP at solving the original benchmark of Chao et al.~\cite{Chao96}. In turn, the results obtained by our implementations of B-B\&C and CPA while solving the new instances (with a non-empty set of mandatory vertices) are discussed in Section~\ref{s_results_stop}.

\subsection{Results for TOP instances}
\label{s_results_top}

Here, we study the behaviour of CPA at solving the TOP benchmark of Chao et al.~\cite{Chao96}.
In this sense, we first analyzed the impact of the inequalities discussed in Section~\ref{s_cuts} on the strength of the formulation $\mathcal{F}_2$. To this end, we computed the dual (upper) bounds obtained from adding these inequalities to $\mathcal{L}_2$ (the linear relaxation of $\mathcal{F}_2$) according to five different configurations, as described in Table~\ref{table_cuts_configurations}. Precisely, for each instance and configuration, we solved the cutting-plane phase described in Figure~\ref{fig_cutting-plane} while considering only the types of inequalities of the corresponding configuration.
\begin{table}[!ht]
\center
\caption{Configurations of valid inequalities.}
\label{table_cuts_configurations}
\begin{tabular}{lccc}
\toprule
& \multicolumn{3}{c}{Inequalities}\\
\cmidrule{2-4}
{\bf Configuration} & GCCs & CCs & LCIs \\
\midrule
1 & $\times$ &&\\
2 & & $\times$ & \\
3  & & & $\times$\\
4  & $\times$ & $\times$ & \\
5  & $\times$ & $\times$ & $\times$\\
\bottomrule
\end{tabular}
\end{table}

The results are detailed in Table~\ref{table_results_lps_top}.
The first column displays the name of each instance set. Then, for each configuration of inequalities, we give the average and the standard deviation (over all the instances in each set) of the percentage bound improvements obtained from the addition of the corresponding inequalities. Without loss of generality, given an instance, its percentage improvement in a configuration $i \in \{1,2,3,4,5\}$ is given by $100 \cdot \frac{UB_{LP} - UB_{i}}{UB_{LP}}$, where $UB_{LP}$ denotes the bound provided by $\mathcal{L}_2$, and $UB_{i}$ stands for the bound obtained from solving the cutting-plane phase in the configuration $i$.
The last row displays the numerical results while considering the complete benchmark of instances. We remark that these results do not take into account the CPLEX built-in cuts, since only GCCs, CCs and LCIs are separated at the cutting-plane phase.
\afterpage{%
    \clearpage% Flush earlier floats (otherwise order might not be correct)
    \thispagestyle{empty}% empty page style (?
\begin{landscape}
\begin{table}[!ht]
\sisetup{table-format=2.2} 
\setlength\tabcolsep{3pt}
\center
\caption{Percentage dual (upper) bound improvements obtained from adding to $\mathcal{L}_2$ the inequalities of Section~\ref{s_cuts} according to the five configurations in Table~\ref{table_cuts_configurations}. Results for the original benchmark of TOP instances.}
\label{table_results_lps_top}
    \begin{tabular}
    {l
    S[table-format=2.2]
    S[table-format=2.2]
    S[table-format=2.2]
    S[table-format=2.2]
    S[table-format=2.2]
    S[table-format=2.2]
    S[table-format=2.2]
    S[table-format=2.2]
    S[table-format=2.2]
    S[table-format=2.2]
    S[table-format=2.2]
    S[table-format=2.2]
    S[table-format=2.2]
    S[table-format=2.2]
    S[table-format=2.2]
    }
\toprule
& & \multicolumn{14}{c}{Configuration of inequalities}\\
\cmidrule{3-16}
 & & \multicolumn{2}{c}{1 --- GCCs} & & \multicolumn{2}{c}{2 --- CCs} & & \multicolumn{2}{c}{3 --- LCIs} & & \multicolumn{2}{c}{4 --- GCCs \& CCs} & & \multicolumn{2}{c}{5 --- All} \\ 
 \cmidrule{3-4} \cmidrule{6-7} \cmidrule{9-10} \cmidrule{12-13} \cmidrule{15-16}
\textbf{Set }& & {Avg (\%)} & {StDev (\%)} & & {Avg (\%)} & {StDev (\%)} & & {Avg (\%)} & {StDev (\%)} & & {Avg (\%)} & {StDev (\%)} & & {Avg (\%)} & {StDev (\%)} \\
\midrule
1 & & 3.61 & 3.09 & & 4.46 & 3.87 & & 0.83 & 2.03 & & 4.77 & 3.79 & & 5.12 & 3.89\\
2 & & 0.14 & 0.44 & & 0.42 & 1.38 & & 0.77 & 2.38 & & 0.44 & 1.40 & & 1.08 & 2.88\\
3 & & 1.12 & 1.08 & & 2.08 & 1.64 & & 0.62 & 1.01 & & 2.12 & 1.65 & & 3.01 & 2.16\\
4 & & 4.01 & 3.53 & & 3.51 & 3.70 & & 0.01 & 0.02 & & 4.90 & 4.07 & & 4.92 & 4.06\\
5 & & 0.51 & 1.65 & & 0.86 & 1.81 & & 0.18 & 0.65 & & 0.89 & 1.80 & & 1.08 & 1.91\\
6 & & 0.00 & 0.00 & & 0.00 & 0.00 & & 0.04 & 0.10 & & 0.00 & 0.00 & & 0.04 & 0.10\\
7 & & 3.67 & 2.27 & & 5.40 & 4.25 & & 0.32 & 1.10 & & 6.00 & 3.80 & & 6.16 & 3.72\\
\midrule
\textbf{Total} & & 1.98 & 2.73& & 2.54 & 3.43& & 0.37 & 1.25& & 2.90 & 3.59& & 3.21 & 3.68\\
\bottomrule
\end{tabular}
\end{table}
\end{landscape}
\clearpage% Flush earlier floats (otherwise order might not be correct)
}

The results exposed in Table~\ref{table_results_lps_top} indicate that, on average, CCs are the inequalities that strengthen formulation $\mathcal{F}_2$ the most, followed by GCCs and LCIs. The results also suggest that GCCs, CCs and LCIs do not dominate each other. In fact, coupling the three of them always leads to greater or equal average bound improvements than considering each inequality alone. We also point out that, although LCIs alone only provide marginal average improvements on the bounds, coupling them with the other inequalities is still effective. Such behaviour is somehow expected, as the separation of LCIs relies on the quality of the currently available bounds. Then, these LCIs tend to work better once the bounds are already strengthened by other inequalities.

In a second experiment, we evaluated the performance of CPA by comparing the results obtained by the algorithm with the ones of B-B\&C reported in \cite{Bianchessi2017}. To make a fair comparison, we also report the results of our implementation of B-B\&C running within our experimental environment. The results are shown in Table~\ref{table_results_top2}. The first column displays the name of
each instance set, and, for each algorithm, we give four result values described as follows. The first value  corresponds to the number of instances solved at optimality out of the complete instance set. The second one is the average wall-clock processing time (in seconds) spent in solving these instances. Note that this entry only takes into account the instances solved at optimality. The last couple of result values provides the average and the standard deviation (only over the unsolved instances in each set) of the relative optimality gaps obtained by the algorithm. These gaps are given by $\frac{UB-LB}{UB}$, where $LB$ and $UB$ are, respectively, the best lower and upper bounds obtained by the corresponding algorithm for a given instance. Whenever $LB = UB = 0$, the corresponding optimality gap is set to 0\%.
The last row gives the overall results considering the complete benchmark of instances.

\afterpage{%
    \clearpage% Flush earlier floats (otherwise order might not be correct)
    \thispagestyle{empty}% empty page style (?
\begin{landscape}
\begin{table}[!ht]
\sisetup{table-format=3.2} 
\setlength\tabcolsep{4pt}
\center
\caption{Comparison between B-B\&C and CPA at solving the original benchmark of TOP instances.}
\label{table_results_top2}
\begin{tabular}{lccrrcccrrrccrrr}
\toprule
 & & \multicolumn{9}{c}{B-B\&C (Bianchessi et al.~\cite{Bianchessi2017})} & & \multicolumn{4}{c}{}\\ 
 \cmidrule{3-11}
  & & \multicolumn{4}{c}{Original report} & & \multicolumn{4}{c}{Our implementation} & & \multicolumn{4}{c}{CPA (our algorithm)} \\ 
 \cmidrule{3-6} \cmidrule{8-11} \cmidrule{13-16}
 & & \multicolumn{2}{c}{\textit{solved}} & \multicolumn{2}{c}{\textit{unsolved}} & & \multicolumn{2}{c}{\textit{solved}} & \multicolumn{2}{c}{\textit{unsolved}} & & \multicolumn{2}{c}{\textit{solved}} & \multicolumn{2}{c}{\textit{unsolved}} \\ 
 \cmidrule{3-6} \cmidrule{8-11} \cmidrule{13-16}
& & & & \multicolumn{2}{c}{Gap (\%)} & &&& \multicolumn{2}{c}{Gap (\%)} & & & & \multicolumn{2}{c}{Gap (\%)}\\
\cmidrule{5-6} \cmidrule{10-11} \cmidrule{15-16}
\textbf{Set }& & \#opt/total & Time (s) & Avg & StDev & & \#opt/total & Time (s) & Avg & StDev & & \#opt/total & Time (s) & Avg & StDev\\
\midrule
1 & & \textbf{54/54} & 1.10 & -- & --  & & \textbf{54/54} & 0.70 & -- & -- & & \textbf{54/54} & 1.91 & -- & --\\
2 & & \textbf{33/33} & 0.20 & -- & -- & & \textbf{33/33} & 0.07 & -- & -- & & \textbf{33/33} & 0.13 & -- & --\\
3 & & \textbf{60/60} & 184.90 & -- & --  & & \textbf{60/60} & 109.63 & -- & -- & & \textbf{60/60} & 106.33 & -- & --\\
4 & & 39/60 & 870.40 & 2.29 & -- & & 32/60 & 985.92 & 2.59 & 1.69 & & \textbf{43/60} & 1286.50 & 3.03 & 2.57\\
5 & & 60/78 & 517.90 & 3.49 & -- & & 60/78 & 291.58 & 2.95 & 1.45 & & \textbf{62/78} & 395.45 & 3.01 & 1.78\\
6 & & 36/42 & 22.10 & 1.92 & -- & & 39/42 & 183.95 & 1.95 & 0.57 & & \textbf{42/42} & 262.58 & -- & --\\
7 & & 45/60 & 992.80 & 2.53 & -- & & 32/60 & 446.84 & 2.71 & 1.26 & & \textbf{47/60} & 626.11 & 1.94 & 0.75\\
\midrule
\textbf{Total} & & 327/387& 424.49 & 2.67 & -- & & 310/387 & 248.82 & 2.69 & 1.45& & 341/387 & 371.79 & 2.71 & 1.95\\
\bottomrule
\end{tabular}
\end{table}
\end{landscape}
\clearpage% Flush earlier floats (otherwise order might not be correct)
}
From the results, one may note that the average optimality gaps of the solutions obtained by our implementation of B-B\&C are extremely close to those presented in the original report.
%In fact, the average gaps obtained by our implementation are smaller than or equal to those presented in the original report for all instance sets.
On the other hand, our implementation of B-B\&C solved to optimality significantly less instances than the original report. Particularly, it finds difficulty in closing the gaps of the largest instances (sets 4 and 7). We believe that such behaviour is not only due to the differences in hardware, but also to some specific implementation choices, such as the algorithm adopted to solve the maximum flow problems, the CPLEX solver version and, in special, the values of the parameters discussed in Section~\ref{s_implementation_details}. 
Since the overall performance of our implementation is in accordance with the original report and the latter does not provide all of the implementation details --- in particular, the tolerance and precision values adopted in the separation of GCCs ---, we chose not to address this issue in this study.

%From the results, one may note that the average optimality gaps of the solutions obtained by our implementation of B-B\&C are extremely close to those presented in the original report. On the other hand, our implementation of B-B\&C solved to optimality significantly less instances than the original report. Particularly, it finds difficulty in closing the gaps of the largest instances (sets 4 and 7). We believe that such behaviour is due to a tailing-off phenomenon during the separation of GCCs, which impacts the performance of the algorithm as the size of the instances grows. In this case, this behaviour is likely to be smoothed by tuning the tailing-off tolerance parameter $\epsilon_1$. Since the overall performance of our implementation is in accordance with the original report and the latter does not provide these specific implementation details, we chose not to address this issue in this study.

In any case, the results clearly indicate the superiority of our algorithm (CPA) in solving the original benchmark of TOP instances, even when compared to the original report of B-B\&C. Precisely, our algorithm was able to solve to optimality 31 and 14 more instances than B-B\&C when considering our implementation and the original report in \cite{Bianchessi2017}, respectively. In addition, the average gaps of the solutions (regarding the unsolved instances) provided by CPA are comparable to those of B-B\&C (both in our implementation and in the original report) for all instance sets. We also remark that CPA and B-B\&C present comparable average execution times as well.

As expected, instances whose graphs have greater dimensions are the hardest (sets 4, 5 and 7). We also noticed that instances with greater route duration limits (given by $T$) tend to be more difficult to solve. This is possibly due to the fact that greater limits imply more feasible routes, thus increasing the search space.
On the other hand, the number of vehicles available does not seem to interfere with the difficulty in solving the instances. We believe this is in accordance with the way we model the problem in this work. Precisely, one may note that the size of formulation $\mathcal{F}_2$ (as well as $\mathcal{F}_1$) does not depend on the number of vehicles, as all the routes are implicitly modeled by means of a single commodity.

In Table~\ref{table_results_top}, we summarize, for each instance set, the total number of instances solved to optimality by each exact algorithm in the literature, including our CPA. This table is displayed for completeness purposes, as the differences in hardware and experimental environments are not taken into account. The results for the branch-and-price (B\&P) and the branch-and-cut-and-price (B\&C\&P) algorithms of Keshtkaran et al. \cite{Keshtkaran16} are presented separately. Moreover, the algorithm of Poggi et al. \cite{Poggi10} was omitted due to the lack of complete results in the original report.

\afterpage{%
    \clearpage% Flush earlier floats (otherwise order might not be correct)
    \thispagestyle{empty}% empty page style (?
\begin{landscape}
\begin{table}[ht]
\sisetup{table-format=1.0} 
\setlength\tabcolsep{4pt}
\center
\caption{Total number of instances solved by each exact algorithm in the literature of TOP.}
\label{table_results_top}
\begin{tabular}{lcccccccc}
\toprule
 & \multicolumn{1}{c}{} & \multicolumn{1}{c}{} & \multicolumn{2}{c}{Keshtkaran et al.~\cite{Keshtkaran16}} &  \multicolumn{1}{c}{} & \multicolumn{2}{c}{B-B\&C (Bianchessi et al.~\cite{Bianchessi2017})} &\\ 
 \cmidrule{4-5} \cmidrule{7-8}
 & \multicolumn{1}{c}{Boussier et al.~\cite{Boussier07}} & \multicolumn{1}{c}{Dang et al.~\cite{Dang13b}} & \multicolumn{1}{c}{B\&P} & \multicolumn{1}{c}{B\&C\&P} &
  \multicolumn{1}{c}{El-Hajj et al.~\cite{ElHajj2016}} & \multicolumn{1}{c}{Original report} & \multicolumn{1}{c}{Our implementation} & CPA (ours)\\ 
  \midrule
 {\bf Set } & \#opt/total& \#opt/total & \#opt/total & \#opt/total & \#opt/total & \#opt/total & \#opt/total & \#opt/total\\
\midrule
1 & 51/54 & \textbf{54/54} & \textbf{54/54} & \textbf{54/54} & \textbf{54/54} & \textbf{54/54} & \textbf{54/54} & \textbf{54/54}\\
2 & \textbf{33/33} & \textbf{33/33} & \textbf{33/33} & \textbf{33/33} & \textbf{33/33} & \textbf{33/33} & \textbf{33/33} & \textbf{33/33}\\
3 & 50/60 & \textbf{60/60} & \textbf{60/60} & 51/60 & \textbf{60/60} & \textbf{60/60} & \textbf{60/60} & \textbf{60/60}\\
4 & 25/60 & 22/60 & 20/60 & 22/60 & 30/60 & 39/60 & 32/60 & \textbf{43/60}\\
5 & 48/78 & 44/78 & 60/78 & 59/78 & 54/78 & 60/78 & 60/78 & \textbf{62/78}\\
6 & 36/42 & \textbf{42/42} & 36/42 & 38/42 & \textbf{42/42} & 36/42 & 39/42 & \textbf{42/42}\\
7 & 27/60 & 23/60 & 38/60 & 34/60 & 27/60 & 45/60 & 32/60 & \textbf{47/60}\\
\midrule
{\bf Total} & \multicolumn{1}{c}{270/387} & \multicolumn{1}{c}{278/387} & \multicolumn{1}{c}{301/387} & \multicolumn{1}{c}{291/387} & \multicolumn{1}{c}{300/387} & \multicolumn{1}{c}{327/387} & \multicolumn{1}{c}{310/387} & \multicolumn{1}{c}{\textbf{341/387}}\\
\bottomrule
\end{tabular}
\end{table}
\end{landscape}
\clearpage% Flush earlier floats (otherwise order might not be correct)
}

In total, CPA solved 14 more instances than any previous exact algorithm and was able to prove the optimality of eight TOP instances previously unsolved, which are listed in Table~\ref{table_new_opts}.
In this table, we report, for each instance, the value of the optimal bound found (column ``Opt.'') and the wall-clock time (in seconds) spent by CPA to solve the instance. Additionally, we display the bound provided by the linear relaxation of $\mathcal{F}_2$ (column ``LP''), as well as the (possibly) improved bound obtained after separating GCCs, CCs and LCIs at the root node of the branch-and-bound tree (column ``LP+cuts''). Once again, we highlight that the latter values (``LP+cuts'') do not consider the CPLEX built-in cuts.

\begin{table}[ht]
\center
\caption{New optimal solutions found by CPA at solving the original benchmark of TOP instances.}
\label{table_new_opts}
\begin{tabular}{lrrrrr}
\toprule
 & & \multicolumn{4}{c}{CPA (our algorithm)} \\
\cmidrule{3-6}
Instance & & \multicolumn{1}{c}{Opt.} & \multicolumn{1}{c}{Time (s)} & \multicolumn{1}{c}{LP} & \multicolumn{1}{c}{LP+cuts} \\ 
\midrule
{p5.3.x} & & 1555.00 & 3833.30 & 1591.07 & 1591.07\\
{p4.2.p} & & 1242.00 & 5192.14 & 1306.00 & 1288.49\\
{p4.3.m} & & 1063.00 & 4553.93 & 1220.71 & 1131.82\\
{p4.3.o} & & 1172.00 & 5322.10 & 1287.18 & 1230.14\\
{p4.3.p} & & 1222.00 & 2436.57 & 1300.97 & 1265.48\\
{p4.4.l} & & 880.00 & 2273.87 & 972.42 & 919.95\\
{p7.2.q} & & 1044.00 & 2324.04 & 1129.62 & 1089.30\\
{p7.3.q} & & 987.00 & 3455.92 & 1078.77 & 1022.57\\
\bottomrule
\end{tabular}
\end{table}

\subsection{Results for the new STOP instances}
\label{s_results_stop}

Now, we study the performance of B-B\&C and CPA at solving the new benchmark of STOP instances.
The name of each new STOP instance (set) corresponds to the original name of the TOP instance (set) from which it was generated, followed by the percentage of vertices selected as mandatory (in this case, 5\%).

As for the TOP instances, we first analyzed the impact of the inequalities discussed in Section~\ref{s_cuts} on the strength of formulation $\mathcal{F}_2$. To this end, we computed the dual (upper) bounds obtained from adding these inequalities to $\mathcal{L}_2$ according to the five configurations described in Table~\ref{table_cuts_configurations}. Precisely, for each instance and configuration, we solved the cutting-plane phase described in Figure~\ref{fig_cutting-plane} while considering only the types of inequalities of the corresponding configuration.

The results are detailed in Table~\ref{table_results_lps_stop}.
The first column displays the name of each instance set. Then, for each configuration of inequalities, we give the average and the standard deviation (over all the instances in each set) of the percentage bound improvements obtained from the addition of the corresponding inequalities. Without loss of generality, given an instance, its percentage improvement in a configuration $i \in \{1,2,3,4,5\}$ is given by $100 \cdot \frac{UB_{LP} - UB_{i}}{UB_{LP}}$, where $UB_{LP}$ denotes the bound provided by $\mathcal{L}_2$, and $UB_{i}$ stands for the bound obtained from solving the cutting-plane phase in the configuration $i$. The last row displays the numerical results while taking into account the complete benchmark of instances.

\afterpage{%
    \clearpage% Flush earlier floats (otherwise order might not be correct)
    \thispagestyle{empty}% empty page style (?
\begin{landscape}
\begin{table}[!ht]
%\sisetup{table-format=2.2} 
\setlength\tabcolsep{3pt}
\center
\caption{Percentage dual (upper) bound improvements obtained from adding to $\mathcal{L}_2$ the inequalities of Section~\ref{s_cuts} according to the five configurations in Table~\ref{table_cuts_configurations}. Results for the new STOP instances.}
\label{table_results_lps_stop}
    \begin{tabular}
    {l
    S[table-format=2.2]
    S[table-format=2.2]
    S[table-format=2.2]
    S[table-format=2.2]
    S[table-format=2.2]
    S[table-format=2.2]
    S[table-format=2.2]
    S[table-format=2.2]
    S[table-format=2.2]
    S[table-format=2.2]
    S[table-format=2.2]
    S[table-format=2.2]
    S[table-format=2.2]
    S[table-format=2.2]
    S[table-format=2.2]
    }
\toprule
& & \multicolumn{14}{c}{Configuration of inequalities}\\
\cmidrule{3-16}
 & & \multicolumn{2}{c}{1 --- GCCs} & & \multicolumn{2}{c}{2 --- CCs} & & \multicolumn{2}{c}{3 --- LCIs} & & \multicolumn{2}{c}{4 --- GCCs \& CCs} & & \multicolumn{2}{c}{5 --- All} \\ 
 \cmidrule{3-4} \cmidrule{6-7} \cmidrule{9-10} \cmidrule{12-13} \cmidrule{15-16}
{\bf Set }& & {Avg (\%)} & {StDev (\%)} & & {Avg (\%)} & {StDev (\%)} & & {Avg (\%)} & {StDev (\%)} & & {Avg (\%)} & {StDev (\%)} & & {Avg (\%)} & {StDev (\%)} \\
\midrule
1\_5\% & & 7.32 & 4.67 & & 8.83 & 5.98 & & 1.00 & 3.93 & & 9.26 & 6.01 & & 9.46 & 6.00\\
2\_5\% & & 0.32 & 0.90 & & 0.86 & 2.44 & & 0.51 & 1.43 & & 0.85 & 2.41 & & 0.97 & 2.75\\
3\_5\% & & 1.48 & 1.35 & & 2.71 & 1.90 & & 0.59 & 0.99 & & 2.71 & 1.87 & & 3.22 & 2.05\\
4\_5\% & & 5.94 & 5.62 & & 5.15 & 6.46 & & 0.00 & 0.01 & & 7.27 & 6.62 & & 7.27 & 6.63\\
5\_5\% & & 0.18 & 0.65 & & 0.86 & 1.46 & & 0.02 & 0.10 & & 0.84 & 1.32 & & 0.87 & 1.40\\
6\_5\% & & 0.06 & 0.17 & & 0.69 & 2.24 & & 0.03 & 0.07 & & 0.70 & 2.24 & & 0.73 & 2.23\\
7\_5\% & & 5.96 & 4.24 & & 8.95 & 9.58 & & 0.00 & 0.00 & & 10.16 & 8.85 & & 10.08 & 8.59\\
\midrule
\textbf{Total} & & 3.18 & 4.46& & 4.03 & 5.95& & 0.26 & 1.52& & 4.63 & 6.14& & 4.75 & 6.09\\
\bottomrule
\end{tabular}
\end{table}
\end{landscape}
\clearpage% Flush earlier floats (otherwise order might not be correct)
}

The results exposed in Table~\ref{table_results_lps_stop} follow the same pattern we observed during the resolution of the original TOP instances. Precisely, they indicate that, on average, CCs are the inequalities that strengthen formulation $\mathcal{F}_2$ the most, followed by GCCs and LCIs. Moreover, the results suggest that GCCs, CCs and LCIs do not dominate each other and that coupling the three of them leads to a greater average bound improvement than considering each inequality alone. 

One may notice that, for some instance sets, coupling inequalities gives worse average bound improvements than considering them separately (see, e.g., set 2\_5\% under configuration 4). Such behaviour can be explained by the fact that the classes of inequalities considered are separated heuristically, as discussed by the end of Section~\ref{s_implementation_details}.

Then, we compared the performance of our implementations of B-B\&C and CPA. The results are shown in Table~\ref{table_results_stop2}. The first column displays the name of
each instance set, and, for each algorithm, we give four result values described as follows. The first value  corresponds to the number of instances solved at optimality (or to proven infeasibility) out of the complete instance set. The second one is the average wall-clock processing time (in seconds) spent in solving these instances. Note that this entry only takes into account the instances solved at optimality. The last couple of result values provides the average and the standard deviation (only over the unsolved instances in each set) of the relative optimality gaps obtained by the algorithm. Recall that these gaps are given by $\frac{UB-LB}{UB}$, where $LB$ and $UB$ are, respectively, the best lower and upper bounds obtained by the corresponding algorithm for a given instance.
If, for a given instance, no feasible solution is found within the time limit and its infeasibility is also not proven, the corresponding optimality gap is assumed to be 100\%. Likewise, this gap is set to 0\% whenever the instance is proven to be infeasible. 
The last row gives the overall results considering the complete benchmark of instances.

The results indicate that CPA outperforms B-B\&C in terms of the quality of the solutions obtained when solving the new benchmark of STOP instances. Aside from having the total average gap of the solutions (concerning unsolved instances) smaller than that of B-B\&C, our algorithm was able to solve to optimality 30 more instances than B-B\&C. Although B-B\&C presents smaller average execution times, these values are still close enough to the ones obtained by CPA, as they have a same order of magnitude for most of the instance sets.

Notice that, for the instance set with the greatest dimensions (set 7\_5\%), the standard deviation of the optimality gaps obtained by both algorithms were particularly high. This is partially due to the fact that, for a few instances in this set, the algorithms could neither find feasible solutions nor prove their infeasibility within the time limit, thus implying optimality gaps of 100\% in these cases. In fact, by analyzing the results in a per-instance basis, we observed that both algorithms had difficulty in proving the infeasibility of the new STOP instances when that was the case. On the other hand, from the experiments, we could not conclude whether fixing vertices as mandatory (in the new STOP instances) complicates or favors the solvability of the feasible instances.

In general, the results for the new STOP instances indicate a similar behaviour as the one observed when solving the original TOP instances. Precisely, instances with greater route duration limits tend to be more difficult to be solved by both algorithms, and the number of vehicles available does not seem to interfere with the difficulty in solving the instances.

\begin{table}[!th]
\sisetup{table-format=4.2} 
\setlength\tabcolsep{4pt}
\center
\caption{Comparison between B-B\&C and CPA at solving the new benchmark of STOP instances.}
\label{table_results_stop2}
\begin{tabular}{lccrrrccrrr}
\toprule
  & & \multicolumn{4}{c}{B-B\&C (Our implementation)} & & \multicolumn{4}{c}{CPA (our algorithm)} \\ 
 \cmidrule{3-6} \cmidrule{8-11}
 & & \multicolumn{2}{c}{\textit{solved}} & \multicolumn{2}{c}{\textit{unsolved}} & & \multicolumn{2}{c}{\textit{solved}} & \multicolumn{2}{c}{\textit{unsolved}} \\ 
  \cmidrule{3-6} \cmidrule{8-11}
& & & & \multicolumn{2}{c}{Gap (\%)} & & & & \multicolumn{2}{c}{Gap (\%)}\\
\cmidrule{5-6} \cmidrule{10-11}
\textbf{Set }& & \#opt/total & Time (s) & Avg & StDev & & \#opt/total & Time (s) & Avg & StDev\\
\midrule
1\_5\% & & \textbf{54/54} & 1.03 & -- & -- & & \textbf{54/54} & 1.62 & -- & --\\
2\_5\% & & \textbf{33/33} & 0.03 & -- & -- & & \textbf{33/33} & 0.03 & -- & --\\
3\_5\% & & \textbf{60/60} & 133.68 & -- & -- & & \textbf{60/60} & 130.31 & -- & --\\
4\_5\% & & 30/60 & 639.86 & 3.15 & 2.10 & & \textbf{41/60} & 1085.46 & 3.58 & 2.85\\
5\_5\% & & 62/78 & 170.10 & 2.45 & 0.94 & & \textbf{65/78} & 353.58 & 2.70 & 1.10\\
6\_5\% & & 39/42 & 174.27 & 2.20 & 0.08 & & \textbf{42/42} & 219.66 & -- & --\\
7\_5\% & & 38/60 & 145.89 & 18.61 & 33.28 & & \textbf{51/60} & 790.22 & 13.04 & 32.63\\
\midrule
\textbf{Total} & & 316/387 & 158.73 & 7.75 & 19.70& & \textbf{346/387} & 361.04 & 5.38 & 15.30\\
\bottomrule
\end{tabular}
\end{table}

\section{Concluding remarks}
In this work, we introduced the Steiner Team Orienteering Problem (STOP) and proposed a Cutting-Plane Algorithm (CPA) for it.
The algorithm works by solving a commodity-based compact formulation reinforced by the separation of three families of inequalities, which consist of
some General Connectivity Constraints (GCCs), classical Lifted Cover Inequalities (LCIs) based on dual bounds and a class of Conflict Cuts (CCs). 
To our knowledge, CCs were also introduced in this work. 

Extensive computational experiments showed that CPA is highly competitive in solving a benchmark of TOP instances.
In particular, the algorithm solved, in total, 14 more instances than any other exact algorithm in the literature of TOP. Moreover, our approach was able to find optimal solutions for eight previously unsolved instances. Regarding the new STOP instances introduced in this work, our algorithm solved 30 more instances than a state-of-the-art branch-and-cut algorithm for TOP adapted to STOP.
From the results, we concluded that CPA benefits from both the strength and compactness of the model used as backbone, as well as from the reinforcement provided by the three families of inequalities aforementioned. 

We believe that the mathematical formulation adopted in this work can be successfully applied to model other routing problems, especially the ones in which the routes are subject to distance and/or time constraints. Moreover, the CCs introduced in our work can also help solving other routing problems via cutting-plane schemes. 

In terms of future work directions for STOP, we believe that one can benefit from the compactness of the model proposed here to develop mathematical programming based heuristics (matheuristics) (see, e.g., \cite{Archetti2014}), field almost unexplored in the literature related to STOP. 
Accordingly, future works can also apply matheuristics to improve the convergence of CPA by building feasible integer solutions from fractional ones. In fact, this strategy has been successfully applied to solve a problem  related to STOP, the so-called orienteering arc routing problem \cite{Archetti2016}.

\section*{Acknowledgments}
This work was partially supported by the Brazilian National Council
for Scientific and Technological Development (CNPq), the Foundation
for Support of Research of the State of Minas Gerais, Brazil
(FAPEMIG), and Coordination for the Improvement of Higher Education
Personnel, Brazil (CAPES).

\section*{References}
 \bibliographystyle{elsarticle-num}
 \bibliography{main}

%% else use the following coding to input the bibitems directly in the
%% TeX file.
\clearpage
\appendix
\section{Bound comparison between $\mathcal{F}_1$ and $\mathcal{F}_2$}
\label{appendix_0}

Here, we summarize the results obtained from our experimental analysis of the impact of the valid inequalities (\ref{b109}) and (\ref{csc110}) on the strength of formulations $\mathcal{F}_1$ and $\mathcal{F}_2$, respectively. The results for the original benchmark of TOP instances of Chao et al.~\cite{Chao96} and the new STOP instances introduced in this work are detailed in Tables~\ref{table_results_lps2_top} and~\ref{table_results_lps2_stop}, respectively.
In both tables, the first column displays the name of each instance set. Then, we give the average and the standard deviation (over all the instances in each set) of the percentage bound improvements referred to (\ref{b109}) on $\mathcal{F}_1$ and to (\ref{csc110}) on $\mathcal{F}_2$. In both tables, the last row displays the numerical results while considering the complete benchmark of instances.

Without loss of generality, given an instance, the percentage bound improvement referred to (\ref{b109}) on $\mathcal{F}_1$ is given by $100 \cdot \frac{UB_{\mathcal{L}_1\backslash(\ref{b109})}-UB_{\mathcal{L}_1}}{UB_{\mathcal{L}_1\backslash(\ref{b109})}}$, where $UB_{\mathcal{L}_1}$ denotes the bound provided by $\mathcal{L}_1$ (the linear relaxation of $\mathcal{F}_1$), and $UB_{\mathcal{L}_1\backslash(\ref{b109})}$ stands for the bound obtained from solving $\mathcal{L}_1$ without inequalities (\ref{b109}). Likewise, given an instance, the percentage bound improvement referred to (\ref{csc110}) on $\mathcal{F}_2$ is given by $100 \cdot \frac{UB_{\mathcal{L}_2\backslash(\ref{csc110})}-UB_{\mathcal{L}_2}}{UB_{\mathcal{L}_2\backslash(\ref{csc110})}}$, where $UB_{\mathcal{L}_2}$ denotes the bound provided by $\mathcal{L}_2$ (the linear relaxation of $\mathcal{F}_2$), and $UB_{\mathcal{L}_2\backslash(\ref{csc110})}$ stands for the bound obtained from solving $\mathcal{L}_2$ without inequalities (\ref{csc110}). Recall that, from Theorem~\ref{teo01}, $UB_{\mathcal{L}_1}$ is always equal to $UB_{\mathcal{L}_2}$.

\begin{table}[!ht]
\center
\caption{Impact --- in terms of percentage bound improvement --- of the valid inequalities (\ref{b109}) and (\ref{csc110}) on the linear relaxations of $\mathcal{F}_1$ and $\mathcal{F}_2$, respectively. Results for the original benchmark of TOP instances.}
\label{table_results_lps2_top}
\begin{tabular}
{
l
S[table-format=2.2]
S[table-format=2.2]
S[table-format=2.2]
S[table-format=2.2]
S[table-format=2.2]
S[table-format=2.2]
}
\toprule
 & & \multicolumn{2}{c}{$\mathcal{L}_1$ --- with (\ref{b109})} & & \multicolumn{2}{c}{$\mathcal{L}_2$ --- with (\ref{csc110})} \\ 
 \cmidrule{3-4} \cmidrule{6-7}
{\bf Set }& & {Avg (\%)} & {StDev (\%)} & & {Avg (\%)} & {StDev (\%)}\\
\midrule
1 & & 4.52 & 4.20 & & 2.70 & 2.79\\
2 & & 0.33 & 0.97 & & 0.30 & 0.84\\
3 & & 3.40 & 3.06 & & 2.12 & 1.85\\
4 & & 2.37 & 2.69 & & 2.29 & 2.53\\
5 & & 3.08 & 3.34 & & 1.88 & 2.11\\
6 & & 0.17 & 0.45 & & 0.16 & 0.45\\
7 & & 4.73 & 4.65 & & 3.50 & 3.61\\
\midrule
\textbf{Total} & & 2.93 & 3.60& & 2.03 & 2.56\\
\bottomrule
\end{tabular}
\end{table}

\begin{table}[!ht]
\center
\caption{Impact --- in terms of percentage bound improvement --- of the valid inequalities (\ref{b109}) and (\ref{csc110}) on the linear relaxations of $\mathcal{F}_1$ and $\mathcal{F}_2$, respectively. Results for the new STOP instances.}
\label{table_results_lps2_stop}
\begin{tabular}
{
l
S[table-format=2.2]
S[table-format=2.2]
S[table-format=2.2]
S[table-format=2.2]
S[table-format=2.2]
S[table-format=2.2]
}
\toprule
 & & \multicolumn{2}{c}{$\mathcal{L}_1$ --- with (\ref{b109})} & & \multicolumn{2}{c}{$\mathcal{L}_2$ --- with (\ref{csc110})} \\ 
 \cmidrule{3-4} \cmidrule{6-7}
{\bf Set }& & {Avg (\%)} & {StDev (\%)} & & {Avg (\%)} & {StDev (\%)}\\
\midrule
1\_5\% & & 3.51 & 4.96 & & 2.06 & 3.12\\
2\_5\% & & 0.17 & 0.99 & & 0.15 & 0.85\\
3\_5\% & & 2.30 & 3.05 & & 1.39 & 1.86\\
4\_5\% & & 1.78 & 3.13 & & 1.83 & 3.20\\
5\_5\% & & 2.70 & 6.20 & & 1.89 & 5.11\\
6\_5\% & & 0.23 & 0.67 & & 0.22 & 0.66\\
7\_5\% & & 2.61 & 4.97 & & 2.01 & 3.95\\
\midrule
\textbf{Total} & & 2.11 & 4.36& & 1.52 & 3.40\\
\bottomrule
\end{tabular}
\end{table}

\end{document}